\documentclass[12pt]{article}

\usepackage[inline,shortlabels]{enumitem}
\usepackage{float}
\usepackage{mathtools}
\usepackage{subdepth}
\usepackage{comment}
\usepackage{graphicx}
\usepackage[round]{natbib}
\usepackage[margin=1in]{geometry}
\usepackage{tikz}
\usepackage{colortbl}

\usepackage{xcolor}

\usepackage[linesnumbered]{algorithm2e}
\mathtoolsset{showonlyrefs}

\usetikzlibrary{arrows.meta}
\usetikzlibrary{decorations.markings}
\usetikzlibrary{math}
\usetikzlibrary{bending}
\usetikzlibrary{backgrounds}
\RestyleAlgo{ruled}

\usepackage[english]{babel}
\usepackage{longtable}
\usepackage{nccmath}
\usepackage{color}
\usepackage{mathtools}
\usepackage{amssymb,amsmath,amsthm}
\usepackage{multirow}
\usepackage[titletoc,title]{appendix}
\usepackage{authblk}
\usepackage{bbm}
\usepackage{setspace}
\usepackage{dsfont}
\usepackage[OT1]{fontenc}
\usepackage{subcaption}
\usepackage{refcount}
\usepackage{booktabs}
\graphicspath{ {../simulations/01_init_manuscript/graphs/} }
\usepackage{xr}
\usepackage[colorlinks,citecolor=blue,urlcolor=blue]{hyperref}

\newcommand{\titlepaper}{Computationally and statistically efficient estimation of time-smoothed counterfactual curves}

\date{}
\author[1,*]{Herbert P. Susmann}
\author[2]{Nicholas T. Williams}
\author[1]{Richard Liu}
\author[3]{Jessica G. Young}
\author[1]{Iv\'an D\'iaz}

\affil[1]{\small Division of Biostatistics, Department of Population
  Health, New York University Grossman School of Medicine, New York, USA}
\affil[2]{\small Department of Epidemiology
Mailman School of Public Health, Columbia University}
\affil[3]{\small Department of Population Medicine, Harvard Medical School and Harvard Pilgrim Health Care Institute, Boston, USA}
\affil[*]{\small Corresponding author. Address: 180 Madison Avenue, New York, NY, 10016. Email: \url{susmah01@nyu.edu}}

\newtheorem{theorem}{Theorem}
\AtEndDocument{\refstepcounter{theorem}\label{finalthm}}
\AtEndDocument{\refstepcounter{equation}\label{finaleq}}

{
  \theoremstyle{definition}
  \newtheorem{assumption}{}
}
{
  \theoremstyle{definition}
  
}
{
  \theoremstyle{definition}
  
}

\newtheorem{lemma}{Lemma}
\AtEndDocument{\refstepcounter{lemma}\label{finallemma}}

\renewcommand{\P}{\mathsf{P}}
\renewcommand{\Pr}{\mathbb{P}}

\newcommand{\F}{\mathsf{F}}

\newcommand{\G}{\mathcal{G}}

\newcommand{\Rem}{\mathsf{R}}
\renewcommand{\d}{d}

\newcommand{\m}{\mathsf{m}}
\newcommand{\nb}{N_{[\,]}}

\DeclareMathOperator{\diag}{diag}

\newcommand{\indep}{\mbox{$\perp\!\!\!\perp$}} 

\newcommand{\dd}{\,\mathrm{d}}
\newcommand{\Pn}{\mathsf{P}_n}

\newcommand{\thetasr}{\hat\theta_{\mbox{\scriptsize sr}}}
\newcommand{\thetasdr}{\hat\theta_{\mbox{\scriptsize sdr}}}
 
\newcommand{\one}{\mathds{1}}
 
\newcommand{\prob}{\mathbb{P}}
\newcommand{\E}{\mathsf{E}}
\renewcommand{\P}{\mathsf{P}}

\newcommand{\1}{\mathbbm{1}}
\usepackage{accents}

\renewenvironment{proof}{{\it Proof }}{\qed \\}
\DeclareMathOperator*{\argmin}{\arg\!\min}
\DeclarePairedDelimiterX{\norm}[1]{\lVert}{\rVert}{#1}

\pgfdeclarelayer{background}
\pgfsetlayers{background,main}
\usetikzlibrary{arrows,positioning}
\tikzset{
>=stealth',
punkt/.style={
rectangle,
rounded corners,
draw=black, very thick,
text width=6.5em,
minimum height=2em,
text centered},
pil/.style={
->,
thick,
shorten <=2pt,
shorten >=2pt,}
}
\newcommand{\Vertex}[3]
{\node[minimum width=0.6cm,inner sep=0.05cm] (#2) at (#1) {#3};
}
\newcommand{\Vertexr}[3]
{\node[rectangle, draw, minimum width=0.6cm,inner sep=0.05cm] (#2) at (#1) {#2};
}
\newcommand{\ArrowR}[3]%
{ \begin{pgfonlayer}{background}
\draw[->,#3] (#1) to[bend right=30] (#2);
\end{pgfonlayer}
}
\newcommand{\ArrowLW}[3]%
{ \begin{pgfonlayer}{background}
\draw[->,#3] (#1) to[bend left=30] (#2);
\end{pgfonlayer}
}

\newcommand{\ArrowL}[3]%
{ \begin{pgfonlayer}{background}
    \draw[->,#3] (#1) to[bend left=45] (#2);
  \end{pgfonlayer}
}

\newcommand{\EdgeL}[3]%
{ \begin{pgfonlayer}{background}
\draw[dashed,#3] (#1) to[bend right=-45] (#2);
\end{pgfonlayer}
}

\allowdisplaybreaks
\pgfarrowsdeclare{arcs}{arcs}{...}
{
  \pgfsetdash{}{0pt} 
  \pgfsetroundjoin   
  \pgfsetroundcap    
  \pgfpathmoveto{\pgfpoint{-5pt}{5pt}}
  \pgfpatharc{180}{270}{5pt}
  \pgfpatharc{90}{180}{5pt}
  \pgfusepathqstroke
}
\newcommand{\ArrowB}[3]%
{ \begin{pgfonlayer}{background}
    \draw[|-arcs,line width=0.4mm,shorten <= 0.3cm,shorten >= 0.3cm,#3] (#1) -- +(#2);
  \end{pgfonlayer}
}

\SetKwFunction{Extract}{Extract}
\SetKwFunction{Lag}{Lag}
\SetKwFunction{Leap}{Leap}
\SetKwFunction{CrossFit}{CrossFit} 
\SetKwFunction{SplitData}{SplitData} 
\SetKwFunction{Regress}{Regress}
\SetKwFunction{ComputePseudoOutcome}{ComputePseudoOutcome}
\SetKwFunction{Constrain}{Constrain}
\SetKwFunction{EstimateDensityRatio}{EstimateDensityRatio}
\SetKwFunction{RegressPredict}{RegressAndPredict}
\SetKwFunction{Predict}{Predict} \SetKwFunction{Append}{AppendColumns}
\SetKwFunction{Subset}{Subset} \SetKwFunction{Mean}{Mean}
\SetKwFunction{Var}{Variance} \SetKw{KwRet}{Return}
\SetKwInput{KwInput}{Input} \SetKwProg{Fn}{Function}{:}{}
\SetKwData{fit}{fit}\SetKwData{out}{out}
\newif\ifshowarrows
\pgfkeys{
  /lmtpwide/.is family, /lmtpwide,
  default/.style = {
    tau=4,
    target=4,
    showarrows=true
  },
  tau/.estore in = \lmtpwidetau,
  target/.estore in = \lmtpwidetarget,
  showarrows/.is if = showarrows,
  showarrows=true
}

\newcommand{\lmtpwide}[1]{
    \pgfkeys{/lmtpwide, default, #1}
    \pgfmathsetmacro{\tau}{\lmtpwidetau}
    \pgfmathsetmacro{\target}{\lmtpwidetarget}
    \pgfmathtruncatemacro{\tauminusone}{\tau-1}
    \pgfmathtruncatemacro{\targetminusone}{\target-1}
    
    \foreach \x in {1,...,\tau}{
        \pgfmathsetmacro{\hue}{(\x - 1) / \tau}
        \definecolor{mycolor\x}{hsb}{\hue,0.3,0.7}
        \def\basecolor{mycolor\x}
        
        \ifnum \x>\target
            \edef\basecolor{gray}
            \pgfmathsetmacro{\opacity}{0.40}
        \else
            \pgfmathsetmacro{\opacity}{1}
        \fi
        
        \def\boxtext{}
         \ifnum \x>\targetminusone
            \ifnum \x=\tau
                \edef\boxtext{$Y_\x$}
            \else
                \def\boxtext{$\tilde{Y}_{\tau,\x}$}
            \fi
        \else
            \edef\boxtext{$Z_\x$}
        \fi
        
        \fill[\basecolor!100, opacity=\opacity](\x,0) rectangle ++(0.7,0.7) node[pos=0.5,color=white] {\boxtext};
    }

    \ifshowarrows
        \ifnum \target > 1
            \foreach \x in {1,...,\targetminusone}{
                \draw[<-,>=latex] (\x+0.25,0.75) to[bend left=20] (\target + 0.25,0.75);
            };
        \fi
    \fi
}

\newcommand{\lmtplong}[1]{
    \pgfkeys{/lmtpwide, default, #1}
    \pgfmathsetmacro{\tau}{\lmtpwidetau}
    \pgfmathsetmacro{\target}{\lmtpwidetarget}
    \pgfmathtruncatemacro{\tauminusone}{\tau-1}
    \pgfmathtruncatemacro{\targetminusone}{\target-1}

    \foreach \y in {1,...,\tauminusone}{
        \foreach \x in {1,...,\tau}{
            \pgfmathtruncatemacro{\outcome}{\tau - \y + 1}
            \pgfmathtruncatemacro{\lagged}{\outcome - (\tau - \x)}
            \def\boxtext{}
             \ifnum \x>\targetminusone
                \ifnum \x=\tau
                    \def\boxtext{$Y_\outcome$}
                \else
                    \ifnum \lagged > 0
                        \def\boxtext{$\tilde{Y}_{\outcome,\lagged}$}
                    \else
                        \def\boxtext{}
                    \fi
                \fi
            \else
                \ifnum \lagged > 0
                    \def\boxtext{$Z_\lagged$}
                \else
                    \def\boxtext{}
                \fi
            \fi

            \ifnum \x=\tau
                \pgfmathsetmacro{\hue}{(\outcome - 1) / \tau}
            \else
                \pgfmathsetmacro{\hue}{(\lagged - 1) / \tau}
            \fi
            
            \definecolor{mycolor\x}{hsb}{\hue,0.3,0.7}
            \def\basecolor{mycolor\x}

            \pgfmathsetmacro{\opacity}{1}
            \ifnum \x > \target
                \pgfmathsetmacro{\opacity}{0.4}
                \edef\basecolor{gray}
            \else
                \ifnum \x=\tau
                    \pgfmathsetmacro{\c}{(1 - (\outcome)/6)*100}
                \else
                    \ifnum \lagged > 0
                        \pgfmathsetmacro{\c}{(1 - (\lagged)/6)*100}
                    \else
                        \pgfmathsetmacro{\opacity}{0.4}
                        \edef\basecolor{gray}
                    \fi
                \fi
            \fi

            \ifshowarrows
                \ifnum \y > \targetminusone
                    \pgfmathsetmacro{\opacity}{0.4}
                    \edef\basecolor{gray}
                \fi
            \fi
            
            \fill[\basecolor!100, opacity=\opacity](\x,-1.25 * \y) rectangle ++(0.7,0.7) node[pos=0.5,color=white] {\boxtext};
        }
    }

    \ifshowarrows
        \ifnum \target > 1
            \foreach \y in {1,...,\targetminusone}{
                \foreach \x in {1,...,\targetminusone}{
                    \draw[<-,>=latex] (\x+0.25,-1.25 * \y+0.75) to[bend left=20] (\target + 0.25,-1.25 * \y+0.75);
                };
            }
        \fi
    \fi
}

\title{\titlepaper}

\begin{document}
\maketitle

\begin{abstract}
Longitudinal causal inference is concerned with defining, identifying, and estimating the effect of a time-varying intervention on a time-varying outcome that is indexed by a follow-up time. In an observational study, Robins's generalized g-formula can identify causal effects induced by a broad class of time-varying interventions. Various methods for estimating the generalized g-formula have been posed for different outcome types, such as a failure event indicator \textit{by} a specified time (e.g. mortality by 5 year follow-up), as well as continuous or dichotomous/multi-valued outcomes measures \textit{at} a specified time (e.g. blood pressure in mm/hg or an indicator of high blood pressure at 5-year follow-up). Multiply-robust, data-adaptive estimators leverage flexible nonparametric estimation algorithms while allowing for statistical inference. However, extant methods do not accommodate time-smoothing when multiple outcomes are measured over time, which can lead to substantial loss of precision. We propose a novel multiply-robust estimator of the generalized g-formula that accommodates time-smoothing over numerous available outcome measures. Our approach accommodates any intervention that can be described as a \textit{Longitudinal Modified Treatment Policy}, a flexible class suitable for binary, multi-valued, and continuous longitudinal treatments. Our method produces an estimate of the \textit{effect curve}: the causal effect of the intervention on the outcome at each measurement time, taking into account censoring and non-monotonic outcome missingness patterns. In simulations we find that the proposed algorithm outperforms extant multiply-robust approaches for effect curve estimation in scenarios with high degrees of outcome missingness and when there is strong confounding. We apply the method to study longitudinal effects of union membership on wages. 
\end{abstract}

\section{Introduction}
Under causal assumptions, Robins's generalized g-formula \citep{Robins86} can non-parametrically identify, via only measured study variables in a longitudinal observational study, effects of flexibly/pragmatically defined time-varying treatment interventions on an outcome mean indexed by a subsequent follow-up time. These  include interventions that may depend dynamically, and possibly stochastically, on time-varying past characteristics, including natural treatment values \citep{richardson2013single,young2014identification}.  Such interventions have been referred to as longitudinal modified treatment policies \citep[LMTPs,][]{diaz2023nonparametric}.

When targeting an effect of an LMTP via the generalized g-formula, several types of outcomes may be of interest.  For example, a researcher might be interested in an intervention effect on a survival/failure event indicator \textit{by} a specified time (e.g. mortality by 5 year follow-up), a continuous outcome \textit{at} a specified time  (e.g. blood pressure in mm/hg or at 5-year follow-up), or a discrete outcome \textit{at} a specified time (e.g. an indicator of high blood pressure at 5-year follow-up). Researchers may also wish to estimate a set of causal effects for multiple outcomes indexed by time (e.g. the blood pressure in mm/ht at years 1-5 post follow-up). We refer to such a set of causal effects as an \textit{effect curve}. Regardless of outcome type, estimators of the generalized g-formula that in some way allow \textit{time-smoothing} in all available outcome measures are attractive for improved precision.  

The ability to time-smooth in all available measured outcomes may be particularly important for adequate precision in the context of non-survival outcomes, whether continuous or discrete, in common observational study designs.  For example, in so-called ``clinical cohorts''  obtained from electronic health records \citep{hernanobsplans}, outcomes like blood pressure or weight change may be repeatedly measured over a follow-up period yet exhibit substantial (non-monotonically patterned and informative) missingness at any given point in time. In such settings, there may be relatively few individuals with a measure of blood pressure at e.g. exactly 10 months post-baseline relative to the baseline sample size. In this case, an estimator that can only use outcome measures at 10 month, ignoring all available blood pressure measures at 9 and 11 months, will be generally less precise than one that can ``borrow'' information on outcomes at these ``close'' time points -- in other words, that can smooth over time. 

For survival and failure event outcomes, various estimators of the generalized g-formula that time-smooth over available outcomes have been posed. For example, implementations of the parametric g-formula \citep{Robins86,robins2004effects}, a fully parametric estimator of the generalized g-formula that includes a parametric pooled in time outcome hazard model conditional on past treatment and confounders, is available in the \texttt{gformula} R package \citep{mcgrath2020gformula}. Inverse probability (IP) weighted estimators have also been proposed that smooth in a marginal structural model for the time-varying outcome hazards and standardize with respect to baseline covariates to obtain marginal cumulative risk estimates \citep{cain2010start, young2018inverse}.  Extending previous work on time-smoothed repeated outcomes models for informative outcome missingness and/or time-varying treatments \citep{robins1995analysis,hernanrepeated,hu2019causal} to accommodate causal effects of generalized treatment strategies, \cite{mcgrathrepeated} recently posed an IP weighted estimator of the generalized g-formula for such repeated outcomes that allows time-smoothing via a parametric model for the generalized g-formula itself. However, these existing approaches to time-smoothing in estimation of the generalized g-formula have the disadvantage of relying on parametric model assumptions, both for nuisance parameter estimation (i.e., control of baseline and time-varying confounders and selection factors for censoring and outcome measurement) but also for time-smoothing in the outcome for improved precision.  That is, these previous proposals generally pay a bias price for the precision gained by time smoothing. This bias is problematic as it does not dissipate with increasing sample size, leading to confidence intervals that have zero coverage probability at large sample sizes. 

A variety of non-parametric strategies have been developed for estimating causal target parameters, including targeted minimum loss-based estimation, one-step estimation, and double/debiased Machine Learning, among others \citep{pfanzagl1982contributions, vanderLaanRose11, chernozhukov2016double}. For the generalized g-formula, \cite{diaz2023nonparametric} proposed sequentially robust non-parametric estimators for longitudinal causal effects defined as LMTPs, a flexible class handling binary, time-to-event, multi-valued, and continuous longitudinal treatments. \cite{shahu2025estimating} also consider estimating effect curves induced by LMTP interventions. However, these previous estimators do not accommodate time-smoothing. In this work, we propose a time-smoothed, non-parametric sequentially doubly robust (SDR) estimation approach for longitudinal causal effects defined via LMTPs. Our approach avoids the bias introduced by parametric approaches while nevertheless benefiting from the gain of precision possible with time-smoothing. Furthermore, unlike existing estimators which would require multiple runs of the estimation algorithm, the proposed algorithm is computationally efficient, allowing for practical estimation of an effect curve comprising causal effects at multiple time points in a single run.

We prove that the resulting estimator is asymptotically normal and efficient, achieving the semi-parametric efficiency bound for the longitudinal curve parameter. Furthermore, the algorithm incorporates several additional improvements. First, the effect curve algorithm incorporates outcome missingness arising from longitudinal censoring as well as missingness due to sporadic measurement. Second, the proposed algorithm incorporates isotonic regression to stabilize estimation of the double-robust transformations, which ensures that the pseudo-regressions involved stay in their parameter space. We show via simulations that this reduces the variance of the estimator in finite samples, and prove theoretically that it retains all the desirable properties of the SDR estimator. Isotonic regression for stabilizing the SDR algorithm is generally applicable to other estimators of this type, and therefore is of wider interest beyond our specific longitudinal setting \citep{pmlr-v202-van-der-laan23a}. Finally, we propose a method for forming uniform confidence bands for the full curve based on the multiplier bootstrap. 

We begin by introducing notation, defining the \textit{effect curve} target causal parameter, and stating the causal identification results in Section~\ref{sec:notation}. We propose a computationally efficient algorithm for estimating effect curves and performing statistical inference in Section~\ref{sec:algorithm}. Simulations are presented in Section~\ref{sec:simulations}. An application to estimating the effect of union membership on wages in a longitudinal setting is presented in Section~\ref{sec:application}.

\section{Notation, target causal parameter, and identification}
\label{sec:notation}
First, we introduce the longitudinal data structure that we use to define the effect curve target parameter. This data structure extends that of \cite{diaz2023nonparametric}, with the key difference that the outcome is now time-varying. Formally, let $X=(L_1, A_1 \ldots, L_{\tau}, A_{\tau}, L_{\tau+1})$, denote the observed data, where $L_t$ is a vector of time-varying variables that includes a time-varying outcome $Y_t$. The variable $A_t=(Z_t, R_t, C_t)$ is a vector containing a treatment variable $Z_t$, which may be discrete or continuous, an indicator of loss-to-follow-up $C_t$ equal to one if the unit remains in the study at time $t + 1$ and equal to zero otherwise, and an indicator $R_t$ that the outcome was measured at time $t+1$. For both $R_t$ and $C_t$, it is important to keep in mind that they refer to missingness of $Y_{t+1}$. This setup distinguishes between two types of missing data, which are handled differently: missing data due to loss-to-follow-up, via the indicator $C_t$, and missing data due to sporadic outcome data collection, via the indicator $R_t$. The outcome may be numerical or an indicator of not experiencing an event of interest by time $t$, in which case we augment the data with a binary variable  $N_t$ indicating whether the patient is at risk of the outcome at time $t$, i.e., whether the outcome did not occur prior to time $t$. For a numerical outcome we let $N_t=1$ for all~$t$.

\subsection{Notation} Following convention in the longitudinal causal inference literature, we use overbars to denote the history of a variable up to and including time $t$; e.g. $\bar L_t=(L_1,\ldots, L_t)$, and use underbars to denote the future including $t$; e.g. $\underline L_t = (L_t,\ldots, L_{\tau+1})$. We also use $H_t=(\bar L_t, \bar A_{t-1})$ to denote the history of all data measured up to right before $A_t$. Let $X_1, \ldots, X_n$ denote a sample of i.i.d. observations with
$X_i\sim \P$. Let $\P f = \int f(x)\dd \P(x)$ for a given function
$f(x)$. We use $\Pn$ to denote the empirical distribution of
$X_1,\ldots\,X_n$, and assume $\P$ is an element of the nonparametric
statistical model defined as all continuous densities on $X$ with
respect to a dominating measure $\nu$. We let $E$ denote the
expectation with respect to $P$; i.e.,
$E\{f(X)\} = \int f(x)\dd P(x)$, and let $\mathbb{E}$ denote expectation w.r.t the distribution of an i.i.d. sequence $X_1,\ldots,X_n$; i.e., $\mathbb{E}\{f(X_1,\ldots,X_n)\} = \int f(x_1,\ldots,x_n)\prod_{i=1}^n\dd P(x_i)$. Let $||f||^2$ denote the
$L_2(\P)$ norm $\int f^2(x)\dd\P(x)$. Calligraphic font is used to denote the support of a random variable, e.g.,
$\mathcal A_t$ denotes the support of $A_t$. By convention, variables with an index $t\leq 0$ are
defined as the null set, expectations conditioning on a null set are
marginal, products of the type $\prod_{t=k}^{k-1}b_t$ and
$\prod_{t=0}^0b_t$ are equal to one, and sums of the type
$\sum_{t=k}^{k-1}b_t$ and $\sum_{t=0}^0b_t$ are equal to zero. For
vectors $u$ and $v$, $u \le v$ denotes point-wise inequalities.

\subsection{Effect curves} We are concerned with the definition and estimation of the causal effect of an
intervention on the treatment process $\bar Z$ on the time-varying outcome 
$Y_{t}$. The interventions on the treatment process are defined in terms of
longitudinal modified treatment policies, which are
hypothetical interventions where the treatment is assigned a random
variable $Z_t^\d$ which may depend on the natural value of treatment at time $t$
as explained below. An intervention that sets the treatments up to time $t-1$ to
$\bar Z_{t-1}^\d$ generates a counterfactual variable $Z_t(\bar Z_{t-1}^\d)$,
which is referred to as the \textit{natural value of treatment}, and represents the value of treatment that would have been observed at time $t$ under an intervention carried out up until time $t-1$ but discontinued thereafter. An intervention on all the treatment and censoring variables up to $t-1$, together with an intervention $R_{t-1}=1$, yields a counterfactual outcome process $Y_{t}(\bar Z^\d_t)$. Causal effects are defined as contrasts between the counterfactual expectation processes $\theta(t + 1) = \E[Y_{t + 1}(\bar Z^\d_{t+1})]$ for $t \in \{1, \dots, \tau \}$ implied by different interventions $\d$. We refer to the set of counterfactual expectations $(\theta(t + 1) : t \in \{ 1, \dots, \tau \})$ as a \textit{curve}. The contrast between two curves is referred to as an \textit{effect curve}. 

We focus on causal effects defined by a user-given
function $\d(z_t, h_t,\varepsilon_t)$ that maps a given treatment
value $z_t$ (i.e., the ``natural value of treatment'' at time $t$) and a history $h_t$ into a new treatment value.  The function $\d$ is also allowed to depend on a
random variable $\varepsilon_t$, drawn independently across units and independently of all data, and with a 
distribution not dependent on $\P$. 
For fixed values $\bar z_t$, and $\bar l_t$, we
recursively define $z_t^\d=\d(z_t, h^\d_t, \varepsilon_t)$, where
$h^\d_t=(\bar z_{t-1}^\d, \bar l_t)$. The
interventions that we study are thus recursively defined as $Z^\d_t = \d(Z_t(\bar
Z_{t-1}^\d), H_t(\bar Z_{t-1}^\d), \varepsilon_t)$ for a user-given
function $\d$, where we let $Z_1(\bar
Z_{0}^\d)=Z_1$ and $H_1(\bar Z_{0}^\d)=L_1$.

The counterfactual process $\theta(t)$ is identifiable under assumptions given elsewhere (see Assumptions 1-3 of \cite{diaz2023nonparametric} and the assumptions given in \cite{richardson2013single, young2014identification}). For $t\in\{1,\ldots, \tau\}$, let 
\begin{align}
    m_{t+1,t}(z_t, h_t) = E[Y_{t+1}\mid C_t= N_t=R_t =1, Z_t = z_t, H_t=h_t]. 
\end{align}
For $s < t$,
define
\begin{align}
    \label{eq:1}
    m_{t+1, s}(z_s, h_s) = E[N_{s+1}\times m_{t+1, s+1}(d(Z_{s+1}, H_{s+1}), H_{s+1})\mid
  C_s= N_s=1, Z_s = z_s, H_s=h_s].  
\end{align}
Then, under assumptions, $\theta(t+1)$ is identified as
\begin{align}
    \label{eq:identification-result}
  \theta(t + 1) = \E[m_{t+1,1}(d(Z_1, H_1), H_1)],
\end{align}
where $H_1=L_1$. It is helpful when reading the notation to keep in mind that for $m_{t+1, s}$, the first subscript $t+1$ refers to the target outcome time point, and the second subscript index $s$ refers to the first time point in the sequence of regressions. In addition, observe that the \textit{difference} between $t+1$ and $s$ indicates the lag between the final target outcome and the first time point in the sequence of regressions. For example, when the difference between $t+1$ and $s$ is $1$, then the sequential regression $m_{t+1, s}$ is the regression of the outcome $Y_{t+1}$ on the variables one time-point previously.

\section{Statistically and computationally efficient estimation}
\label{sec:algorithm}
We will work our way up to our proposed algorithm in multiple steps. First, we review a standard approach based on applying a sequential regression estimator multiple times, once for each time-varying outcome. This approach is a straightforward application of the identification result \eqref{eq:1}. Next, we introduce the key innovation of our computationally efficient algorithm: by pooling sets of sequential regressions, we are able to significantly reduce the number of regressions required to estimate a full curve. The second algorithm, which we call a \textit{time-smoothed} sequential regression algorithm, applies this principle. Finally, the third algorithm modifies the approach to use sequential double-robust transformations which leads to favorable statistical properties, including multiple-robustness and asymptotic normality under data-adaptive (e.g., machine learning) estimation of the nuisance parameters involved.

\subsection{Standard sequential regression estimator}
The definition of $m_{t+1,s}$ in the identification result \eqref{eq:1} requires estimating a series of
sequential regressions for $s\leq t\in\{1,\ldots,\tau\}$. The standard sequential regression estimator estimates each of these regressions separately. To illustrate, consider the case of $\tau = 3$, in which there are four time-varying outcomes $Y_1$, $Y_2$, $Y_3$, and $Y_4$. The corresponding target parameters are $\theta(2)$, $\theta(3)$, and $\theta(4)$. Each of these parameters could be estimated separately with a standard sequential regression estimator. In each case, a sequence of regressions is estimated starting at $m_{t+1, t}$ and working backward to $m_{t+1, 1}$, then evaluating the identification result \eqref{eq:identification-result}. Figure~\ref{fig:standard-example} illustrates how estimation of $\theta(4)$ proceeds. Another way to illustrate all of required regressions for the full curve starts by seeing that the parameters $m_{t+1,s}$ can be arranged in a
lower diagonal matrix, since $s\leq t$. Each column corresponds to the regressions necessary to estimate one of the parameters. To estimate $\theta(4)$, we start at the diagonal element of the upper left column and work down:
\begin{align}
\label{eq:standard-estimator-matrix}
\underbrace{\begin{bmatrix}
\mathbf{m_{4, 3}} &  & \\
m_{4, 2} & m_{3, 2} & \\
m_{4, 1} & m_{3, 1} & m_{2, 1}\\
\end{bmatrix}}_{\text{Step 1}}\quad\rightarrow\quad \underbrace{\begin{bmatrix}
m_{4, 3} &  & \\
\mathbf{m_{4, 2}} & m_{3, 2} & \\
m_{4, 1} & m_{3, 1} & m_{2, 1}\\
\end{bmatrix}}_{\text{Step 2}}\quad\rightarrow\quad \underbrace{\begin{bmatrix}
m_{4, 3} &  & \\
m_{4, 2} & m_{3, 2} & \\
\mathbf{m_{4, 1}} & m_{3, 1} & m_{2, 1}\\
\end{bmatrix}}_{\text{Step 3}}. 
\end{align}
Estimating $\theta(3)$ and $\theta(2)$ would proceed similarly, estimating the regressions in each of the remaining columns in order. In total, standard methods would require $\tau(\tau - 1) / 2 = 6$ separate regressions to estimate the curve. Such an algorithm is shown in Algorithm~\ref{algo:g-comp-standard}. The algorithm takes as input a ``wide" form dataset in which each observation contains the longitudinal observation from one unit:
\begin{align}
    (L_{i, 1}, A_{i, 1}, \dots, L_{i, \tau}, A_{i, \tau}, L_{i, \tau + 1}, :i\in\{1,\ldots,n\});
\end{align}
recall that each $L_{i, t}$ includes a time-varying outcome $Y_{i, t}$ for $t = 1, \dots, \tau + 1$. 

\begin{algorithm}[H]
    \label{algo:g-comp-standard}
  \caption{Sequential regression g-computation estimator}
  {\small
  \For{$l = 1,\ldots,\tau$} {
    $Z_t^d\gets d(Z_t, H_t)$\;
  }
    
  \For{$l = 1,\ldots,\tau$}
  {
    $\hat m_{l+1,l}\gets \Regress(Y_{l+1}\sim (Z_l, H_l),
    \text{subset} = \{C_l=N_l= R_l=1\})$\;
    $\tilde Y_{l+1,l}\gets \Predict(\hat m_{l+1,l}, \text{data} = \{Z_l^d, H_l\}))$\;
    \uIf{$l > 1$}
    {
      \For{$s = l - 1, \ldots, 1$}
      {
        $\hat m_{t+1,s}\gets \Regress(N_{s+1}\times \tilde Y_{t+1, s+1}\sim ( Z_s, H_s),
        \text{subset} = \{C_s=N_s=1\})$\;
        $\tilde Y_{t+1,s}\gets \Predict(\hat m_{t+1,s}, \text{data} =  \{Z_s^d, H_s\})$\; 
    }
  }
  
  $\thetasr(l+1)\gets \Mean(\tilde Y_{l+1,1})$\;
  }%
  }
\end{algorithm}

\begin{figure}
    \centering
    \begin{tikzpicture}[remember picture]
        \foreach \step in {1,...,3}{
        \pgfmathtruncatemacro{\stepminusone}{\step-1}
        \pgfmathsetmacro{\ycm}{-2.5*\stepminusone}
        \pgfmathsetmacro{\y}{-2.5*\stepminusone + 1}
        \pgfmathtruncatemacro{\target}{4 - \step + 1}
        \pgfmathtruncatemacro{\targetminusone}{4 - \step}

        \ifthenelse{\step > 1}{
            \node[align=left,anchor=south west] at (0, \y) {\textbf{Step \step:} regress $\tilde{Y}_{4,\target}$ on history,};
        }{
            \node[align=left,anchor=south west] at (0, \y) {\textbf{Step \step:} regress $Y_{\target}$ on history,};
        };
        
        \begin{scope}[yshift=\ycm cm,xshift=-0.5cm]
            \lmtpwide{showarrows=true,tau=4,target=\target}
        \end{scope}
        
        \node[align=left,anchor=south west] at (6.5, \y) {predict $\tilde{Y}_{4,\targetminusone}$.};
        \begin{scope}[yshift=\ycm cm,xshift=6cm]
            \lmtpwide{tau=4,target=\targetminusone,showarrows=false}
        \end{scope}
    }
    \node[anchor=south west,align=left] at (0, -6.5) {\textbf{Step 4:} produce final estimate $\hat{\theta}(4) = \frac{1}{n} \sum_{i=1}^n[\tilde{Y}_{4, 1}].$};
    \end{tikzpicture}
    \caption{Illustration of sequential regression algorithm with $\tau = 3$ for the counterfactual mean of the outcome $Y_4$ under a longitudinal modified treatment policy. To estimate a complete longitudinal curve, the same algorithm can be applied $\tau$ times treating respectively $Y_2$, $Y_3$, $Y_4$ as the outcome of interest, yielding the set of effect estimates $\{ \hat{\theta}(2), \hat{\theta}(3), \hat{\theta}(4) \}$.}
    \label{fig:standard-example}
\end{figure}
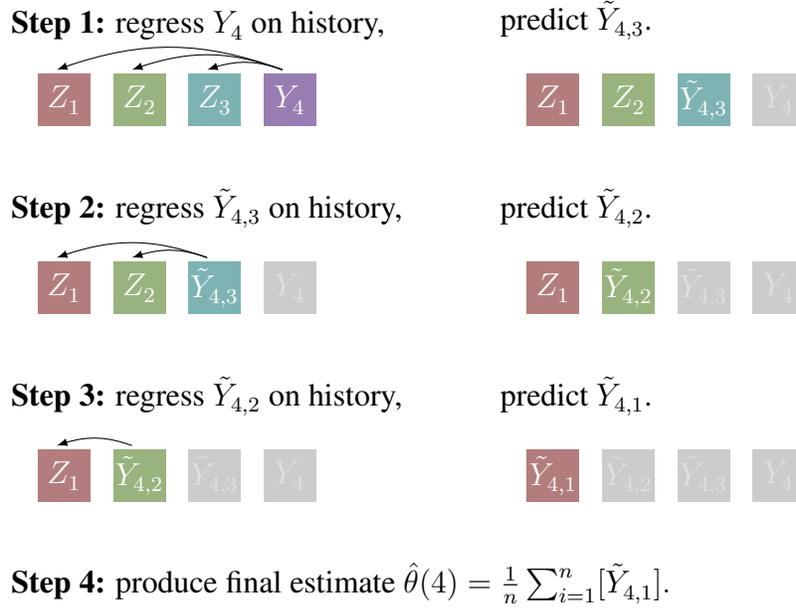

\subsection{Time-smoothed sequential regression estimator}
We propose a modification of the standard sequential regression estimator which works by pooling sets of sequential regressions. The key is to arrange the required sequential regressions in a lower diagonal matrix, as in \eqref{eq:standard-estimator-matrix}, and then estimate the regressions along the diagonals (and subdiagonals) in a single pooled regression. To see why this is desirable, note that the diagonal elements of the matrix correspond to regressions with the same lag between the target outcome $t+1$ and the index $s$. It is perhaps simplest to see the intuition by examining the first diagonal, which corresponds to regressions of the outcome $Y_{t+1}$ against the covariate history at time $t$, for all $t \in \{1, \dots, \tau \}$. The idea is to pool the outcome variables $Y_2$, $Y_3$, and $Y_4$ and the lagged variables $H_1$, $H_2$, and $H_3$, respectively, allowing the regression to detect similar structure in the relationships between the outcome and lagged variables at each time step. Importantly, we also include the variable $t$ in these regressions such that time-dependent relationships between $Y_{t+1}$ and the lagged variables can still be captured as an interaction. An example illustrating the steps of the algorithm is shown in Figure~\ref{fig:smoothed-example}. 

In more detail, the algorithm proceeds by filling in this matrix starting with the
diagonal elements $m_{t+1,t}$, and then proceeding with the
subdiagonal $m_{t+1, t-1}$, the second subdiagonal $m_{t+1, t-2}$, and so
on. This is illustrated in below for the case of $\tau=3$. The algorithm pools data used in the regressions in the diagonal of the matrix, and estimates all diagonal elements with a single regression. This allows for smoothing on the second index of $m_{t+1,s}$.
\begin{equation}
    \underbrace{\begin{bmatrix}
    \mathbf{m_{4, 3}} &  & \\
    m_{4, 2} & \mathbf{m_{3, 2}} & \\
    m_{4, 1} & m_{3, 1} & \mathbf{m_{2, 1}}\\
    \end{bmatrix}}_{\text{Step 1}}\quad\rightarrow\quad
    \underbrace{\begin{bmatrix}
    m_{4, 3} &  & \\
    \mathbf{m_{4, 2}} & m_{3, 2} & \\
    m_{4, 1} & \mathbf{m_{3, 1}} & m_{2, 1}\\
    \end{bmatrix}}_{\text{Step 2}}
    \quad\rightarrow\quad 
    \underbrace{\begin{bmatrix}
    m_{4, 3} &  & \\
    m_{4, 2} & m_{3, 2} & \\
    \mathbf{m_{4, 1}} & m_{3, 1} & m_{2, 1}\\
    \end{bmatrix}}_{\text{Step 3}}
    \label{eq:itera}    
\end{equation}
Fitting these regressions requires considering lagged versions of
the covariates and treatment. Specifically, standard sequential regression approaches conceptualize the recursion (\ref{eq:1}) as a ``wide-form'' dataset in which each outcome at
time $t+1$ is sequentially regressed on its past. For this algorithm, we instead conceptualize it as a ``long-form dataset'' in which every patient
contributes at most $\tau+1$ rows, one for each observed time point. In this dataset, the information for each patient-time is recorded, together with lagged versions of the covariates. Scalable implementation of the estimator may require a Markov independence assumption stating that $(A_{t+1}, H_{t+1})\indep (A_{t-k-1}, H_{t-k-1})\mid (A_{t-k},H_{t-k},\ldots A_t,H_t)$ for some $k$, where we note that setting $k=\tau$ means no independence restriction, and recalling that $A_t=(Z_t,R_t,C_t)$. This assumption allows us to run regressions using only the last $m$ predictors. Specifically, we create the following dataset:
\begin{align}
    (t, Y_{i,t+1}, A_{i,t}, L_{i,t}, A_{i,t-1}, L_{i,t-1}, \ldots, A_{i,t-k}, L_{t-k}, :i\in\{1,\ldots,n\}, t\in\{1,\ldots,T_i\}),
\end{align}
where we denote with $T_i\in\{1,\ldots, \tau\}$ the last time a patient is seen, i.e., the first time $t$ such that $C_{i,t}=0$. We also create an
additional data column indicating the time $t$ of observation of each
record (a patient's record is ended once the patient becomes lost to
follow up). Since we will fit all regressions non-parametrically and include  $t$ as a predictor, without loss of generality we define $A_{i,t-l}, L_{t-l}$ equal to a constant (e.g., zero) if $t-l<0$. We denote this dataset with $\mathcal D_{k}$. The proposed smoothed sequential regression algorithm based is given in Algorithm~\ref{algo:g-comp}.

\begin{algorithm}[H]
  \caption{Time-smoothed sequential regression g-computation estimator}\label{algo:g-comp}
  {\small
    $Z_t^d\gets d(Z_t, H_t)$\;
$\hat m_{t+1,t}\gets \Regress(Y_{t+1}\sim (t, Z_t, H_t),
    \text{subset} = \{C_t=N_t= R_t=1\})$\;
        $\tilde Y_{t+1,t}\gets \Predict(\hat m_{t+1,t}, \text{data} = \{t, Z_t^d, H_t\}))$\;
  \For{$l = 1,\ldots,\tau$}
  {
  $\thetasr(l+1)\gets \Mean(\tilde Y_{l+1,1})$\;
  \uIf{$l <\tau$} {%
  $s\gets t-l$\;
  $\hat m_{t+1,s}\gets \Regress(N_{s+1}\times \tilde Y_{t+1, s+1}\sim (s, Z_s, H_s),
    \text{subset} = \{s > 0, C_s=N_s=1\})$\;
 $\tilde Y_{t+1,s}\gets \Predict(\hat m_{t+1,s}, \text{data} =  \{s, Z_s^d, H_s\}, \text{subset} = \{s > 0\})$\;
 }
  }  
  }%
\end{algorithm}

A simple comparison of Algorithms~\ref{algo:g-comp-standard} and \ref{algo:g-comp} shows that Algorithm~\ref{algo:g-comp-standard} requires $\tau(\tau - 1)/2$ regressions, while Algorithm~\ref{algo:g-comp} requires only $\tau$ regressions, suggesting the latter algorithm may be more computationally efficient. The trade-off is that the regressions necessary in Algorithm~\ref{algo:g-comp} will involve data sets (in ``long" form) with more rows compared to the original algorithm (which uses data in the ``wide" form). 

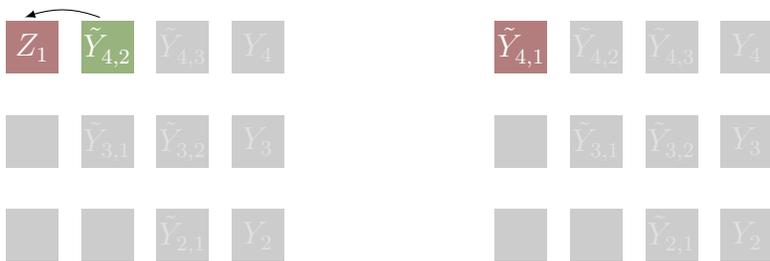
\begin{figure}
    \centering
    \begin{tikzpicture}[]
        \begin{pgfonlayer}{background}
            \draw[->] (4.75, -0.65) to (6.5, -0.65);
        \end{pgfonlayer}

        \node[align=left,anchor=south west] at (0, 0) {\textbf{Step 1:} convert data to augmented ``wide" format.};
        \begin{scope}[yshift=-1cm,xshift=-0.5cm]
            \lmtpwide{tau=4,target=4,showarrows=false}
        \end{scope}

        \begin{scope}[yshift=0.25cm,xshift=6cm]
            \lmtplong{tau=4,,target=4,showarrows=false}
        \end{scope}
        
        \foreach \step in {1,...,3}{
            \pgfmathtruncatemacro{\stepplusone}{\step+1}
            \pgfmathtruncatemacro{\stepminusone}{\step-1}
            \pgfmathtruncatemacro{\stepminustwo}{\step-2}
            \pgfmathsetmacro{\ycm}{-4.8*\step}
            \pgfmathsetmacro{\y}{-4.8*\step}
            \pgfmathtruncatemacro{\target}{4 - \step + 1}
            \pgfmathtruncatemacro{\targetminusone}{4 - \step}o

            \def\regressvar{}
            \def\predictvar{}
            \ifnum \step=1
                \def\regressvar{$Y_{t+1}$}
                \def\predictvar{$\tilde{Y}_{t+1,t}$}
            \fi
            
            \ifnum \step=2
                \def\regressvar{$\tilde{Y}_{t+1, t}$ }
                \def\predictvar{$\tilde{Y}_{t+1, t-1}$}
            \fi

            \ifnum \step>2
                \def\regressvar{$\tilde{Y}_{t+1, t-\stepminustwo}$ }
                \def\predictvar{$\tilde{Y}_{t+1, t-\stepminusone}$ }
            \fi
            
            \node[align=left,anchor=south west] at (0, \y) {\textbf{Step \stepplusone:} regress \regressvar on history,};
            \begin{scope}[yshift=\ycm cm,xshift=-0.5cm]
                \lmtplong{tau=4,target=\target,showarrows=true}
            \end{scope}
            
            \node[align=left,anchor=south west] at (6.5, \y) {predict \predictvar,};
            \begin{scope}[yshift=\ycm cm,xshift=6cm]
                \lmtplong{tau=4,target=\targetminusone,showarrows=false}
            \end{scope}
            
            \node[anchor=south west,align=left] at (11.5, \y) {set $\hat{\theta}(\stepplusone) = \frac{1}{n} \sum_{i=1}^n[\tilde{Y}_{\stepplusone, 1}].$};
        }
    \end{tikzpicture}
    \caption{Illustration of time-smoothed sequential regression algorithm with $\tau = 3$ for estimating the counterfactual mean of $Y_2$, $Y_3$, and $Y_4$ under a modified treatment policy. The algorithm yields estimates $\hat{\theta}(2)$, $\hat{\theta}(3)$, and $\hat{\theta}$(4), which together comprise what we call a longitudinal curve. }
    \label{fig:smoothed-example}
\end{figure}

\subsection{Time-smoothed sequentially doubly robust estimator}
\label{sec:sdr-curve}
While the time-smoothed sequential regression estimator (Algorithm~\ref{algo:g-comp-standard}) incorporates the key innovation of pooling sets of sequential regressions, its statistical properties, such as its sampling distribution are unknown. In this section, we form a time-smoothed sequential doubly robust (SDR) algorithm using doubly-robust transformations to yield an improved estimator that we show is asymptotically normal and efficient under weak conditions on nuisance parameter convergence rates.

The SDR algorithm requires estimation of an additional set of nuisance parameters characterizing the longitudinal treatment assignment process. Formally, for a categorical exposure $Z_t$ we let 
\[g^d_{Z,t}(z_t\mid h_t)=\sum_{z'_t} \one\{d(z_t', h_t) = z_t\}g_{Z,t}(z_t'\mid h_t).\]
denote the probability that $d(Z_t,H_t)=z_t$ conditional on
$\{H_t=h_t,C_{t-1}=1\}$, where we let $g_Z(z_t\mid h_t)$ denote the probability that $Z_t=z_t$ conditional on
$\{H_t=h_t,C_{t-1}=1\}$. If the exposure is continuous, we require the following assumption, originally introduced by \cite{Haneuse2013}:
\begin{assumption}[Piecewise smooth invertibility for continuous exposures]\label{ass:inv}
  For each $h_t$, assume that the support of $Z_t$ conditional on
  $H_t=h_t$ may be partitioned into subintervals
  ${\cal I}_{t,j}(h_t):j = 1, \ldots, J_t(h_t)$ such that
  $d(z_t, h_t)$ is equal to some $d_j(z_t, h_t)$ in
  ${\cal I}_{t,j}(h_t)$ and $d_j(\cdot,h_t)$ has inverse function
  $b_j(\cdot, h_t)$ with derivative $b_j'(\cdot, h_t)$ with
  respect to $a_t$.
\end{assumption}
Under \ref{ass:inv}, the function defined as 
\begin{equation}\label{eq:gdelta}
  g_{Z,t}^d(z_t \mid h_t) =
  \sum_{j=1}^{J_t(h_t)} \one_{t, j} \{b_j(z_t, h_t), h_t\} g_{Z,t}\{b_j(z_t, h_t)\mid h_t\}
  |b_j'(z_t,h_t)|,
\end{equation}
is the probability density function of $d(Z_t,H_t)$ conditional on
$H_t=h_t$ evaluated at $z_t$, where $g_{Z,t}(z_t\mid h_t)$ denotes the probability density function of $Z_t$ conditional on
$\{H_t=h_t,C_{t-1}=1\}$. Define the following probability reweighting function
\[w_{t,s}(A_s,
  H_s)=r_{Z,s}(Z_s,H_s)\frac{I(C_s=1)}{g_{C,s}(Z_s, H_s)}\left[\frac{I(R_s=1)}{g_{R,s}(Z_s,H_s)}\right]^{I(s=t)},\]
    where $r_Z$ is the density ratio
    \[r_{Z,s}(Z_s,H_s)=\frac{g^d_{Z,s}(Z_s\mid
    H_s)}{g_{Z,s}(Z_s\mid H_s)},\]
$g_{C,t}(z_t,h_t) = P(C_t=1\mid Z_t=z_t,H_t=h_t)$, and
$g_{R,t}(z_t,h_t) = P(R_t=1\mid Z_t=z_t,H_t=h_t, C_{t-1}=1)$.  Define the nuisance parameter
$\eta_{t+1} = (w_{t,1},m_{t+1, 1}, \ldots, w_{t,t},m_{t+1, t})$. 

The foundation of the SDR algorithm lies in defining a data transformation that is doubly-robust in a suitable sense. Specifically, for $s=1,\ldots, t$, define the following data transformation, which is similar in form to the efficient influence function:
\[\varphi_{t+1,s}:x\mapsto \sum_{k=s}^t\left(\prod_{l=s}^kw_{t,l}(a_l,
  h_l)\right)\{n_{k+1}\cdot m_{t+1, k+1}(z_{k+1}^d, h_{k+1}) - m_{t+1, k}(z_{k}, h_{k})\}
  + m_{t+1, s}(z_{s}^d, h_{s}),\]
where we have denoted $z_s^d=d(z_s, h_s)$ and $m_{t+1, t+1}=Y_{t+1}$. In prior work \cite[see][]{diaz2024causal,diaz2023nonparametric,luedtke2017sequential,rotnitzky2017multiply} it has been shown that this
data transformation is doubly robust in the sense that
\[E[N_{s+1}\times \varphi_{t+1,s+1}(X;\eta_t')\mid C_{s}= N_{s}=1, Z_{s} = z_{s}, H_{s}=h_{s}] = m_{t+1, s}(z_{s},
  h_{s})\]
for all $s<t$ whenever $\eta_{t+1}'$ is such that, for each $k<t$, we have either $r_{t,k}'=r_{t,k}$, or $m_{t+1,k}'=m_{t+1,k}$. This is desirable because it implies that the expected value of the transformation equals the sequential regression $m_{t+1,s}$ even if some of the nuisance parameters are estimated inconsistently. 

These considerations lead to the following estimation algorithm, in which $\eta_{t+1}'$ is replaced for an estimate. This algorithm is identical to Algorithm \ref{algo:g-comp} with the difference that we have use the pseudo-outcome $\varphi_{t+1,s+1}(X;\eta_t)$ to obtain an estimate of $m_{t+1, s}(z_s,h_s)$. The other important difference is that all regressions should be cross fitted. Under conditions, the above result guarantees that this estimate of $m_{t+1, s}$ is sequentially doubly robust, therefore leading to sequentially doubly robust estimators for the curve $t\mapsto \theta(t)$.

A weakness of the original SDR algorithm is that the estimated sequential regressions $\hat{m}_{t+1, s}$ are not constrained to lie in the support of the outcomes. For example, if the outcomes are binary, there is no guarantee that the range of $\hat{m}_{t+1, s}$ will fall between $0$ and $1$. This may lead to higher variance in the final estimates or point estimates that fall outside the known parameter space. To address this issue, we propose constraining the sequential regressions at each iteration by projecting them onto the space of functions that satisfy a given constraint using isotonic regression. Suppose we have an initial estimate $\hat{m}_{t+1, s}$ of the sequential regression $m_{t+1, s}$ formed by regressing the estimated pseudo-outcomes $\hat{\varphi}_{t+1, s}$ on the suitable set of covariates. We solve the following optimization problem
\begin{align}
    \hat g = \argmin_{g \in \G}\sum_{i  = 1}^n[\hat{\varphi}_{i, t+1,s} - g \circ \hat{m}_{t+1, s}(Z_i)]^2,
\end{align}
where $\G$ is the space of non-decreasing monotone functions $g : \mathbb{R} \mapsto [0,1]$.
An updated estimate of the sequential regression guaranteed to lie in the space of functions with range in $[0,1]$ can then be formed by
\begin{align}
    \tilde{m} (z)= \hat g \circ \hat{m}(z),
\end{align}
where the notation $g \circ f$ denotes function composition.
Practically speaking, the above constraint can be applied with straightforward application of isotonic regression.

\begin{algorithm}[H]
  \caption{Pooled sequentially doubly robust estimator using time-smoothing}\label{algo:sdr}
  {\small
    $Z_{t}^d\gets d(Z_{t}, H_{t})$\;
    $\mathcal V\gets \SplitData(n)$\;
$\hat m_{t+1,t}\gets \CrossFit(Y_{t+1}\sim (t, Z_{t}, H_{t}), \text{subset} = \{C_{t}=N_{t}= R_{t}=1\}, \mathcal V)$\;
    $\hat g_{C,t}\gets \CrossFit(C_t\sim (t, Z_t, H_t), \mathcal V)$\;
        $\hat g_{R,t}\gets \CrossFit(R_t\sim (t, Z_t, H_t), \text{subset} = \{C_{t-1}=1\}, \mathcal V)$\;
    $\hat r_{Z,t}\gets \EstimateDensityRatio(Z_t, H_t, \text{subset} = \{C_{t-1}=1\}, \mathcal V)$\;
        $\hat \varphi_{t+1,t}\gets \ComputePseudoOutcome(\hat m_{t+1,t}, \hat g_{C,t}, \hat g_{R,t}, \hat r_{Z,t})$\;
  \For{$l = 1,\ldots,\tau$}
  {
  $\thetasdr(l+1)\gets \Mean(\hat \varphi_{l+1,1})$\;
   \uIf{$l <\tau$} {%
  $s\gets t-l$\;
  $\hat m_{t+1,s}\gets \CrossFit(N_{s+1}\times \hat \varphi_{t+1,s+1}\sim (s, Z_s, H_s),
    \text{subset} = \{s > 0, C_s=N_s=1\}, \mathcal V)$\;
    $\tilde \m_{t+1,s} \gets \Constrain(\hat \m_{t+1,s})$\;
    $\hat \varphi_{t+1,s}\gets \ComputePseudoOutcome(\tilde m_{t+1,s}, \hat g_{C,s}, \hat g_{R,s}, \hat r_{Z,s})$\;
 }
  }  
}%
\end{algorithm}


The use of the double-robust transformation yields an estimator that converges to a normal distribution centered on the true parameter value and with efficient variance. 
\begin{theorem}[Weak convergence of SDR estimator]\label{theo:weak} Define the data-dependent parameter
\[\check m_{t+1,s}(z_s, h_s)=\E[N_{s+1}\times \varphi_{t+1,s+1}(X;\hat \eta_t)\mid C_{s}= N_{s}=1, Z_{s} = z_{s}, H_{s}=h_{s}].\]
    Assume that, for each $t,s$ $||\hat w_{t,s} - w_{t,s}||\times ||\hat m_{t+1,s} - \check m_{t+1,s}||=o_P(n^{-1/2})$. Assume also that $P(w_{t,s}(A_s, H_s) < c) = P(\hat w_{t,s}(A_s, H_s) < c)=1$ for some $c<\infty$. Then, for the vector $\thetasdr = (\thetasdr(2), \ldots, \thetasdr(\tau+1))$ we have
    \[\sqrt{n}(\thetasdr - \theta)\rightsquigarrow N(0,\Sigma),\]
    where the $(t,s)$ entry of $\Sigma$ is $\mathrm{Cov}[\varphi_{t+1, 1}(X;\eta_t), \varphi_{s+1, 1}(X;\eta_s)]$. In addition, $\Sigma$ is local asymptotic minimax
  efficiency bound for estimation of $\theta = (\theta(2), \ldots, \theta(\tau+1))$ in the sense that, for
  any estimator sequence $\hat\theta_n$:
  \[\inf_{\delta
      >0}\liminf_{n\to\infty}\sup_{Q:V(Q-P)<\delta}n\E\{\hat\theta_n
    - \theta(Q)\}^2\geq \diag\{\Sigma_P\},\] where
  $V(\cdot)$ is the variation norm, $\E$ denotes expectation,
  and $\geq$ denotes element-wise inequality. We added indices $P$
  and $Q$ to emphasize sampling under $P$ or $Q$, and used notation
  $\theta(Q)$ to denote the parameter computed at an arbitrary
  distribution.
\end{theorem}
The proof is given in the appendix. We can use this theorem to construct pointwise confidence intervals and uniform confidence bands on the function $\theta(t)$. A pointwise $(1-\alpha)\times100\%$ confidence interval for $\hat{\theta}_{\mathrm{sdr},t}$, $t \in \{2, \dots, \tau + 1\}$ is given by 
\begin{align}
    \hat{\theta}_{sdr, t} \pm q_{1-\alpha} \times \sqrt{\hat{\sigma}^2 \slash n},
\end{align}
where $\hat\sigma^2(t)$ denotes the empirical
variance of $\varphi_{t+1, 1}(X;\hat \eta_t)$. 

To form a uniform confidence band, first find a value $c_\alpha$ such that
$\hat \rho(c_\alpha) = 1 - \alpha$, where $\hat\rho$ is a function
satisfying
\begin{equation}\label{eq:apmb}
  \hat\rho(s) = \prob\left(\max_{t} \bigg|
    \frac{\thetasdr(t) - \theta(t)}{\hat\sigma(t) /
    \sqrt{n}}\bigg| \leq s \right) + o_\P(1).
\end{equation}
Confidence bands can then be computed as $\thetasdr(t) \pm
n^{-1/2}c_\alpha\hat\sigma(t)$. To approximate the function $\hat\rho(t)$, we use the multiplier
bootstrap \citep{gine1984some,vanderVaart&Wellner96,chernozhukov2013gaussian,
  belloni2018inference}, which has the advantage that it does not require the evaluation of
large covariance matrices nor integration of multivariate normal distributions, and therefore is  more computationally efficient and convenient than approximating (\ref{eq:apmb}) directly. 

The multiplier bootstrap approximates the distribution of
the $\max$ in (\ref{eq:apmb}) with the maximum of the process
$${\mathbb{M}}(t) = \frac{1}{\sqrt{n}}\sum_{i = 1}^n \frac{\xi_i
  \{\varphi_{t+1,1}(X_i;\hat\eta_t) - \thetasdr(t)\}}
  {\hat\sigma(t)},$$
where randomness is introduced through sampling the multipliers
$(\xi_1, \ldots, \xi_n)$, despite the process being conditional on the observed
data $X_1, \ldots, X_n$. The multiplier variables are i.i.d.~with mean zero and
unit variance, and are drawn independently from the sample. Typical choices are
Rademacher ($\prob(\xi = -1) = \prob(\xi=1)=0.5$) or Gaussian multipliers. Under
the assumptions of Theorem~\ref{theo:weak}, plus uniform consistency of
$\hat\sigma(t)$, it can be shown that (\ref{eq:apmb}) holds for
$$\hat\rho(s) = \prob \left(\max_{t} \big|{\mathbb{M}}(t)
  \big|\leq s\,\bigg|\, X_1, \ldots, X_n\right).$$
As a consequence, computation of the critical value requires simulation of a large number of realizations of the multipliers.

\section{Simulation Studies}
\label{sec:simulations}

Simulation studies were conducted on a shared high-performance computing cluster. In comparing runtimes of alternative algorithms, all algorithms were executed on compute nodes configured with 2x 2.4GHZ 20-core Skylake 6148 processors and 12x 32GB DDR4 memory. The SDR curve algorithm is available in the open-source \texttt{lmtp} \texttt{R} package at \url{https://github.com/nt-williams/lmtp/tree/curve} \citep{williams2023lmtp}. Reproduction materials for the simulation studies is available at \url{https://github.com/herbps10/effect\_curve\_paper}. 

\subsection{Simulation Study 1}
For the first simulation study, we extend the simulation data-generating process of \cite{diaz2023nonparametric} to include time-varying outcomes with sporadic missingness. The data-generating process for the time-varying covariates and treatments are as in \cite{diaz2023nonparametric}:
\begin{align}
    L_1 &\sim \mathrm{Categorical}\left(0.5, 0.25, 0.25\right), \\
    A_1 \mid L_1 &\sim \mathrm{Binomial}\left(5, 0.5 \times \1(L_1 > 1) + 0.1 \times \1(L_1 > 2) \right), \\
    L_t \mid \left( \bar{A}_{t-1}, \bar{L}_{t-1} \right) &\sim \mathrm{Bernoulli}\left( \mathrm{logit}^{-1}\left( -0.3 L_{t-1} + 0.5 A_{t-1} \right) \right), \text{ for } t \in \{2, 3, 4\}, \\
    A_t \mid \left( \bar{A}_{t-1}, \bar{L}_t \right) &\sim \begin{cases}
        \mathrm{Binomial}\left( 5,  \mathrm{plogis}\left( -2 + 1 / (1 + 2L_t + A_{t-1} \right) \right) \text{ for } t \in \{2, 3\}, \\
        \mathrm{Binomial}\left( 5, \mathrm{plogis}\left( 1 + L_t - 3 A_{t-1} \right) \right) \text{ for } t \in \{ 4 \}, \\
    \end{cases},
\end{align}
where $X \sim \mathrm{Categorical}(p_1, p_2, p_3)$ denotes the categorical distribution with $P(X = i) = p_i$ for $i \in \{1, 2, 3\}$, $X \sim \mathrm{Bernoulli}(p)$ is the Bernoulli distribution with probability $p$, and $\mathrm{Binomial}(n, p)$ is the Binomial distribution with $n$ trials and probability of success $p$. 
Time-varying outcomes $Y_t$ and binary sporadic missingness indicators $R_t$ are simulated according to
\begin{align}
    Y_t \mid \left( \bar{A}_t, \bar{L}_t \right) &\sim \mathrm{Bernoulli}\left( \mathrm{logit}^{-1}\left( -2 + 1/(1 - 1.2A_t - 0.3L_t \right) \right) \text{ for } t \in \{2, \dots, 5 \}, \\
    R_t \mid \left( \bar{A}_t, \bar{L}_t \right) &\sim \mathrm{Bernoulli}\left( 1 - \mathrm{logit}^{-1}\left( \mathrm{logit}(\alpha) + 2\times\1(L_{t-1} = 1) - 1 \right) \right) \text{ for } t \in \{2, \dots, 5 \}.
\end{align}
The simulation parameter $\alpha$ controls the base level of missingness in the outcome at each time point, with higher levels of $\alpha$ implying higher probability of missing outcome data. 

We generated $N = 500$ datasets for every combination of sample sizes $n \in \{ 500, 1000, 2500, 5000 \}$ and sporadic outcome missingness base probabilities $\alpha \in \{0, 0.5, 0.8 \}$. For each simulated dataset, we compared the proposed time-smoothed sequential doubly robust estimator to a benchmark. The benchmark approach was to apply the standard sequential doubly-robust estimator for a single outcome multiple times, once for each of the time-varying outcomes. (We note that this benchmark was implemented naively, in the sense that all nuisance parameters  were reestimated when computing the effect for each of the time-varying outcomes. A more computationally efficient alternative may be to estimate all the required density ratio nuisance parameters in a first pass, and then re-use them for estimating each of the $\tau$ effects). Both the time-smoothed and benchmark approaches were used to produce point estimates and pointwise 95\% confidence intervals for each of the causal effects that comprise the target curve parameter. We also produced uniform 95\% confidence bands for the estimates from both approaches using the described multiplier bootstrap method with $1000$ bootstrap draws. For each algorithm, all nuisance parameters were estimated using an ensemble of gradient boosting machines (\texttt{lightgbm}; \citealt{shi2024lightgbm}) with 25, 50, and 100 boosting iterations. 

We compared the proposed method to the benchmark in terms of median error (ME), median absolute error (MAE) of the point estimates of the estimated curves, averaging across time-points. We also report the average empirical coverage of the pointwise confidence intervals and the empirical coverage of the uniform confidence bands. To compare the computational efficiency of the methods, we report the relative wall-clock run time of our proposed method relative to the benchmark.

The results of the simulation study are shown in Table~\ref{tab:simulation-study-1}. In all simulations, the average wall-clock runtime of the proposed SDR curve estimator was less than that of the benchmark. For low and moderate levels of sporadic missingness ($\alpha = 0$ and $\alpha = 0.5$) the statistical properties of the SDR curve estimator and the benchmark method are similar in terms of median error and absolute error; however, for high levels of missingness ($\alpha = 0.8$) our proposed algorithm had lower variance than the benchmark, suggesting that time-smoothing yields finite-sample efficiency gains in this setting. In all scenarios and for both algorithms, empirical coverage of the pointwise confidence intervals and uniform confidence bands were near the nominal 95\% level; detailed results are included in the Supplemental Material. 

\begin{table}[]
    \centering
    \begin{tabular}{|lrrrrrr|}
    \hline
    & & \multicolumn{2}{c}{ME $\times 100$} & \multicolumn{2}{c}{MAE $\times 100$} & Relative \\
    $\alpha$ & $n$ & SDR Curve & Benchmark & SDR Curve & Benchmark & Runtime \\
    \hline
    0 & 500 & 0.08 & \bf{0.02} & \bf{1.43} & 1.45 & 0.44\\
     & 1000 & 0.21 & \bf{0.18} & \bf{0.98} & 1.00 & 0.45\\
     & 2500 & 0.22 & \bf{0.19} & \bf{0.64} & 0.64 & 0.46\\
     & 5000 & 0.18 & \bf{0.16} & \bf{0.46} & 0.46 & 0.49\\
    0.5 & 500 & \bf{0.11} & 0.15 & \bf{2.39} & 2.61 & 0.58\\
     & 1000 & \bf{0.04} & 0.05 & \bf{1.65} & 1.90 & 0.60\\
     & 2500 & \bf{0.16} & 0.18 & \bf{1.04} & 1.09 & 0.61\\
     & 5000 & 0.15 & \bf{0.11} & \bf{0.72} & 0.73 & 0.63\\
    0.8 & 500 & -0.69 & \bf{-0.45} & \bf{4.84} & 5.90 & 0.58\\
     & 1000 & \bf{0.06} & 0.13 & \bf{3.18} & 4.08 & 0.59\\
     & 2500 & \bf{-0.06} & -0.10 & \bf{1.93} & 2.24 & 0.62\\
     & 5000 & \bf{0.13} & 0.13 & \bf{1.23} & 1.33 & 0.63\\
    \hline
    \end{tabular}
    \caption{Results of Simulation Study~1 comparing the proposed time-smoothed SDR algorithm (SDR curve; \ref{sec:sdr-curve}) against the benchmark in terms of median error (ME), median absolute error (MAE), and relative runtime. The median error and median absolute error that are closest to zero for each simulation setting are bolded.}
    \label{tab:simulation-study-1}
\end{table}

\subsection{Simulation Study 2}

The second simulation study is designed to investigate the effect of the novel isotonic constraint step in the pooled SDR algorithm. The simulation data-generating process includes two baseline covariates $W$ and $X$ both drawn from a standard normal distribution. Then the following time-varying covariates, treatment variables, and outcomes are generated:
\begin{align}
    L_t &\sim \mathrm{Binomial}(0.5), \text{ for } t \in \{1, \dots, 4 \}, \\
    A_t &\sim \mathrm{Binomial}\left(\mathrm{logit}^{-1}(\alpha(W + X)) \right), \text{ for } t \in \{1, \dots, 4 \} \\
    Y_t &\sim \mathrm{Binomial}\left(\mathrm{logit}^{-1}(-3 + W + A_{t-1} \times X) \right), \text{ for } t \in \{2, \dots, 5 \},
\end{align}
where the parameter $\alpha$ controls the strength of the confounding effect of the baseline covariates.

We generated $N = 500$ datasets for every combination of sample sizes $n \in \{ 500, 1000, 2500, 5000 \}$ and confounding strengths $\alpha \in \{0, 1.5, 3 \}$. In the case of binary outcomes, the time-smoothed sequential doubly-robust estimator estimator proposed in Section~\ref{sec:sdr-curve} includes a step to constrain the range of the estimated sequential regressions to fall in $[0, 1]$ (line 13, Algorithm~\ref{algo:sdr}). We applied this algorithm to each of the simulated datasets, in addition to a simple alternative ``unconstrained" algorithm that does not include this constraint step. We also applied the benchmark estimator described in the previous simulation study that applies the standard SDR estimator for a single outcome successively to each of the time-varying outcomes. Nuisance parameters were estimated as in the previous simulation study. 

The results of Simulation Study 2 are shown in Table~\ref{tab:simulation-study-2}, which compares the three algorithms in terms of their median error and median absolute error. When there is no confounding ($\alpha = 0$) the three methods perform similarly in terms of median error and absolute error. As confounding increases, the SDR curve algorithm has lower median absolute error than the alternatives, suggesting that in this scenario the inclusion of the constraint step stabilizes the variance of the estimator. 

\begin{table}[]
    \centering
    \begin{tabular}{|lrrrrrrr|}
    \hline
    & & \multicolumn{3}{c}{ME $\times 100$} & \multicolumn{3}{c|}{MAE $\times 100$} \\
    & & Unconstrained & & & Unconstrained & & \\
    $\alpha$ & $n$ & SDR Curve & SDR Curve & Benchmark & SDR Curve  & SDR Curve & Benchmark \\
    \hline
    0 & 500 & -0.35 & \bf{-0.32} & -0.49 & 2.65 & \bf{2.22} & 2.70\\
     & 1000 & -0.20 & \bf{-0.07} & -0.20 & 1.96 & \bf{1.61} & 1.84\\
     & 2500 & -0.05 & \bf{0.02} & -0.07 & 1.14 & \bf{1.00} & 1.10\\
     & 5000 & -0.04 & \bf{0.01} & -0.05 & 0.68 & \bf{0.64} & 0.66\\
    1.5 & 500 & \bf{-0.69} & -1.01 & -0.81 & 5.04 & \bf{4.17} & 5.07\\
     & 1000 & -1.01 & -1.17 & \bf{-0.98} & 7.42 & \bf{5.03} & 6.24\\
     & 2500 & \bf{-1.36} & -1.36 & -1.52 & 14.89 & \bf{6.44} & 11.10\\
     & 5000 & -0.44 & -0.50 & \bf{-0.42} & 6.64 & \bf{3.54} & 5.24\\
    3 & 500 & \bf{0.08} & -0.94 & 0.44 & 6.30 & \bf{4.25} & 4.90\\
     & 1000 & 0.71 & -0.87 & \bf{0.46} & 9.53 & \bf{5.34} & 6.84\\
     & 2500 & 0.29 & \bf{-0.21} & 1.07 & 10.14 & \bf{4.12} & 6.60\\
     & 5000 & 0.79 & \bf{0.10} & 0.72 & 5.02 & \bf{2.52} & 3.79\\
    \hline
    \end{tabular}
    \caption{Results of Simulation Study 2 comparing the proposed time-smoothed SDR algorithm (SDR Curve; \ref{sec:sdr-curve}), the time-smoothed SDR algorithm with the isotonic constraint step removed (Unconstrained SDR Curve); and the benchmark in terms of median error (ME) and median absolute error (MAE). The median error and median absolute error closest to zero for each simulation setting are bolded.}
    \label{tab:simulation-study-2}
\end{table}

\section{Application}
\label{sec:application}
To illustrate our methods we revisit a study of the effects of union membership on wages  \citep{vella1998wages}. The data are sourced from the National Longitudinal Survey (Youth Sample) and comprise a sample of males who completed their education by 1980 and who were working full-time. The sample were followed from 1980 to 1987. The exposure of interest was union membership (defined as whether the participant indicated their salary was determined by a collective bargaining agreement). The analysis dataset used by \citep{vella1998wages} is publicly available as the \texttt{wagepan} dataset of the \texttt{wooldridge} package \citep{shea2024wooldridge}. The analysis dataset ($n = 545$) includes 3 baseline covariates (including race/ethnicity indicators and years of schooling) and 27 time-varying covariates (including 9 indicators of occupation type). The exposure variable is a binary indicator of whether wages were reported as set by a collective bargaining agreeement; for simplicity, we refer to the exposure as ``union'' vs. ``non-union''. The time-varying outcome is the reported hourly wage in US dollars. 

We estimated the effect of eight alternative modified treatment policies. For each year from 1980 to 1987, we define an MTP that aligns the exposure status of all participants in that year to the ``union''. In all other years than the index year, the exposure status is not intervened on (that is, the exposure remains at its natural value). For the index year 1983, for example, the MTP intervention induces the counterfactual ``what would the population mean log-wage had been from 1984-1987, possibly contrary to fact, the entire population had their wages set by a collective bargaining agreement in 1983?'' We applied our proposed longitudinal causal effect curve algorithm to estimate the counterfactual population mean wage under each MTP. Nuisance parameters were estimated using an ensemble of learners \citep{vdl2007superlearner} including generalized linear models and random forests (\texttt{ranger}; \citealt{wright2017ranger}). 

The results of the analysis are shown in Figure~\ref{fig:application-results}. For the years 1980-1983, the union intervention had a statistically significant positive effect on wages. In years after the intervention year, in which union status was not intervened on (taking its natural value), the effect of the earlier intervention attenuated. For example, for the 1980 intervention, the effect on wages remained in 1981, and was statistically insignificant for all later years. Taken together, the results indicate a union premium for union membership in the first six years of the study period (1980-1985), yet found no effect in the latter four years. 

\begin{figure}
    \centering
    \includegraphics[width=1\linewidth]{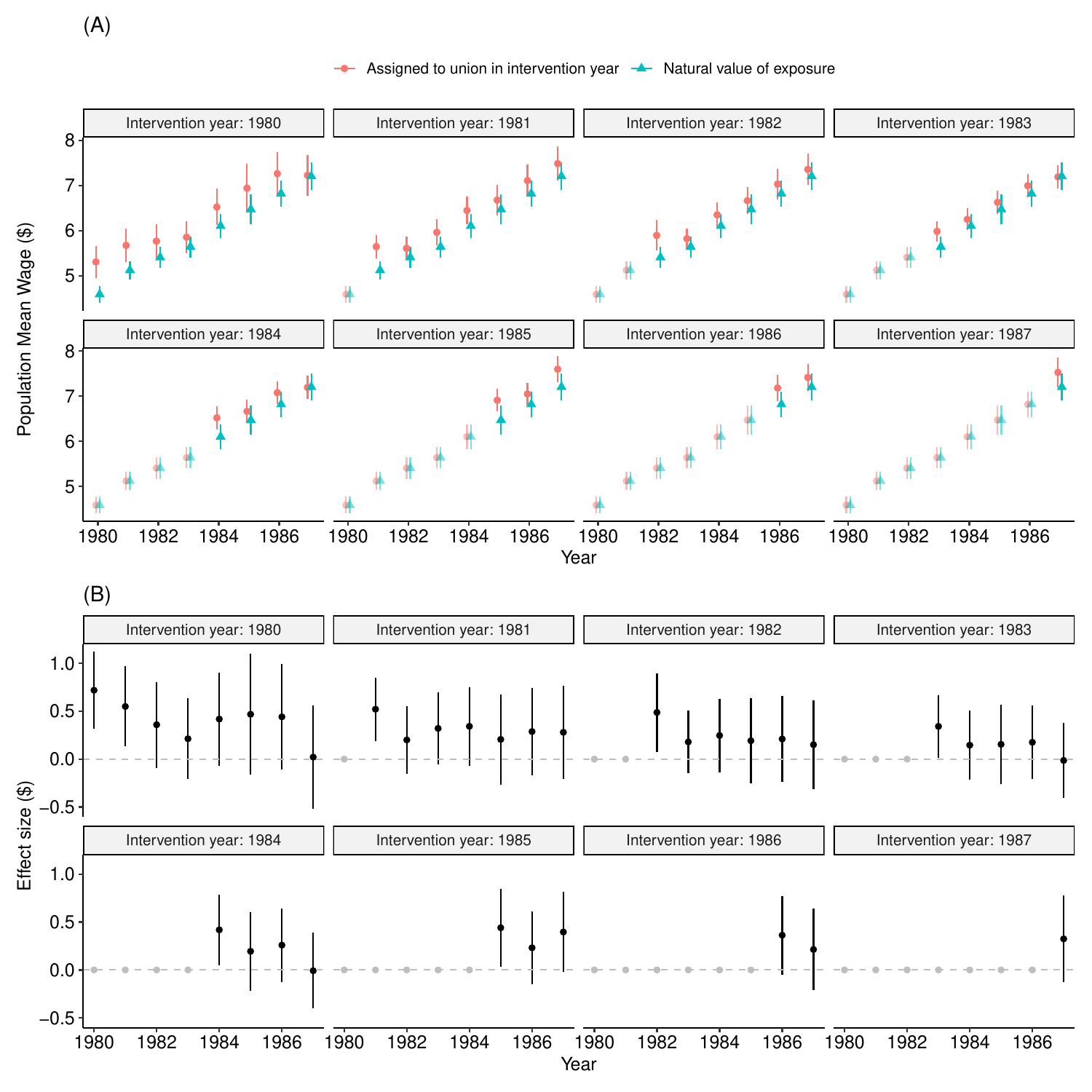}
    \caption{A: estimated population mean wages (\$) and 95\% confidence intervals under hypothetical interventions in which all are assigned to union status in an intervention year (red) vs. union status is left at its natural value (blue). B: point estimates and pointwise 95\% confidence intervals for difference between intervention and non-intervention counterfactual population mean wages. }
    \label{fig:application-results}
\end{figure}

\section{Discussion}
\label{sec:discussion}

Our principal contribution is a approach for estimating the effect of a longitudinal intervention, defined via an LMTP, on a time-varying outcome. Our method exhibits statistical benefits due to its use of time-smoothing in estimating the required sequential regressions. This is particularly relevant when there are high degrees of outcome missingness. Prior methods essentially stratify the required regressions by time; when there is high missingness, this leads to smaller effective sample sizes and higher variance. Our proposed algorithm, on the other hand, pools information across time points, yielding better predictions. This phenomenon is illustrated in Simulation Study 1, in which our proposed approach exhibits smaller finite-sample variance than the benchmark when missingness is high and sample size is low. Asymptotically, we expect the two methods to perform similarly; this is confirmed in the simulations, where the two methods perform similarly at large sample sizes across all missingness levels. 

Our proposed algorithm is also computationally efficient, making practical the estimation of causal effect curves for longitudinal data sets with many measurement times. Compared to a naive application of the sequential regression estimator, which requires estimating $O(\tau^2)$ sequential regressions, our proposal only requires $O(\tau)$ regressions. We expect this to significantly reduce the computational time needed to estimate a longitudinal curve; this is borne out in Simulation Study 1, in which the proposed algorithm has a relative runtime of $0.44-0.65$ compared to the naive benchmark method. 

Another contribution of our work that is relevant to SDR style estimators is our strategy for constraining the estimated sequential regressions to respect the bounds of the outcome variable using isotonic regression. In this style of estimator, carefully designed pseudo-outcomes with double-robust properties are regressed on covariates. Particularly in difficult scenarios with high confounding (low practical covariate overlap), the estimated pseudo-outcomes may be highly variable, and subsequent regressions of the pseudo-outcomes may yield predictions that fall outside the bounds of the parameter space. Our solution is to project the regression onto the space of regressions that respect the parameter space via an isotonic calibration approach, and prove that doing so does not effect the asymptotic properties of the overall estimator. In finite samples, Simulation Study 2 illustrates that including this step can yield substantial finite-sample efficiency gains. We expect this approach to be useful as a way to improve the finite-sample performance of other SDR algorithms.

\section*{Data Availability}
The data used for the application are available in the \texttt{wooldridge} package \url{https://cran.r-project.org/web/packages/wooldridge/index.html}. 

\section*{Acknowledgements}
The computational requirements for this work were supported in part by the NYU Langone High Performance Computing (HPC) Core's resources and personnel.

\clearpage

\section*{Supplemental Material}

\setcounter{section}{0}
\renewcommand\thesection{S\arabic{section}}

\section{Calibration of estimates using isotonic regression}
    Let $X_i=(Y_i,Z_i)\sim\P\in \mathcal P$ for $i\in[n]$ denote an i.i.d. sample from variable taking values in a bounded set $\mathcal X=\mathcal Y\times \mathcal Z$. Let $\mathcal V_1,\ldots, \mathcal V_J$ denote a partition of $\{1,\ldots, n\}$, and define $\mathcal T_j=\{1,\ldots, n\}\backslash \mathcal V_j$. Define $j(i)$ as the function that returns the index of the set $\mathcal V_j$ to which observation $i$ belongs. Let $\theta : {\mathcal P}\mapsto \mathcal H$ denote a parameter of interest, where $\mathcal H$ is the Hilbert space of all functions $f: \mathcal Z \mapsto [0,1]$. Let $\eta$ denote a nuisance parameter. For distributions $\P,\F¯\in\mathcal P$, let $\varphi(\cdot;\eta):\mathcal X \mapsto \Re$ denote a function such that $\E_\P[\varphi(X;\eta_\F)\mid Z=z] = \theta_\P(z) + R(z;\eta_\P,\eta_\F)$, for some error $R(z;\eta_\P,\eta_\F)$.   For each $j$, $\hat\eta_{j}$ denote an estimate of the nuisance parameter, and let $\hat\theta_{j}(z)$ denote a preliminary estimate of $\theta_\P(z)$, both trained using only data in $\mathcal T_j$. Let $\tilde \theta_j(z)=\hat g \circ \hat \theta_j(z)$, where $\hat g$ is defined as 
    \[\hat g = \argmin_{g \in \G}\sum_{i  = 1}^n[\varphi(X_i;\hat \eta_{j(i)}) - g \circ \hat\theta_{j(i)}(Z_i)]^2,\]
    with $\G$ denoting the space of non-decreasing monotone functions $g:\Re \mapsto [0,1]$.
\begin{lemma}[Bounds on the error of isotonic calibration of estimates]\label{theo:iso}
Assume there is a constant $M$ such that $\sup_{x\in\mathcal X}\varphi(x,\hat \eta_j)<M$ for all $j$ and $\sup_{x\in\mathcal X}\varphi(x,\eta)<M$. For all $j\in[J]$, we have
\[||\tilde\theta_j - \theta||\lesssim ||\hat\theta_j - \theta|| + O_{\Pr}(n^{-1/3} + ||R(\cdot;\eta,\hat\eta_j)||).\]
\end{lemma}


\begin{proof}
This proofs follows roughly the arguments of the proof of Theorem 4.8 in \cite{pmlr-v202-van-der-laan23a}. Let $\P_j$ denote the empirical distribution of $\{X_i:i\in \mathcal V_j\}$. For any function $f$, we define the  norm $||f||$ as $||f||^2 = \int [f(x,y)]^2\dd\P(x,y)$.  Let \[R(g) = \sum_{j\in[J]}\P[\varphi(\cdot;\eta) - g\circ \hat\theta_j(\cdot)]^2,\] $\check g = \argmin_{g\in \G}R(g)$, and $\check \theta_j = \check g \circ \hat \theta_j$. Let $\hat\varphi_j(\cdot)=\varphi(\cdot;\hat \eta_{j})$ and $\varphi(\cdot)=\varphi(\cdot;\eta)$. For any $\epsilon\in[0,1]$ and $g\in\G$, we have $\check g + \epsilon(g - \check g)\in \G$, and therefore by definition of $\check g$ we have
\begin{equation}\label{eq:optim1}
\lim_{\epsilon\downarrow 0}\frac{R(\check g + \epsilon(g - \check g)) - R(\check g)}{\epsilon}=-2\sum_{j\in [J]} \P[(g - \check g )\circ \hat\theta_j][\varphi - \check g \circ \hat\theta_j]\geq 0. 
\end{equation}
First, note that 
\begin{equation}\label{eq:ineq1}    
\begin{aligned}
    &\sum_{j\in[J]}||(\check g - \hat g)\circ \hat\theta_j||^2 =  \sum_{j\in[J]}\P[(\check g - \hat g)\circ \hat\theta_j][(\check g - \hat g)\circ \hat\theta_j]\\
    &=\sum_{j\in[J]}\P[(\check g - \hat g)\circ \hat\theta_j][\check g\circ \hat\theta_j - \theta] - \sum_{j\in[J]}\P[(\check g - \hat g)\circ \hat\theta_j][\hat g \circ \hat\theta_j - \theta]\\
        &=\sum_{j\in[J]}\P[(\check g - \hat g)\circ \hat\theta_j][\check g\circ \hat\theta_j - \varphi] - \sum_{j\in[J]}\P[(\check g - \hat g)\circ \hat\theta_j][\hat g \circ \hat\theta_j - \varphi]\\
        &\leq - \sum_{j\in[J]}\P[(\check g - \hat g)\circ \hat\theta_j][\hat g \circ \hat\theta_j - \varphi]\\
        &=-\sum_{j\in[J]}\P_j[(\check g - \hat g)\circ \hat\theta_j][\hat g \circ \hat\theta_j - \varphi] + \sum_{j\in[J]}(\P_j- \P)[(\check g - \hat g)\circ \hat\theta_j][\hat g \circ \hat\theta_j - \varphi]
\end{aligned}
\end{equation}
where the third equality follows from $\E[\varphi(X;\eta)\mid Z=z] = \theta(z)$ and the inequality from (\ref{eq:optim1}) with $g=\hat g$. Furthermore, an argument similar to that leading to (\ref{eq:optim1}) yields, for any $g$:
\[-\sum_{j\in[J]}\P_j[(g - \hat g)\circ \hat\theta_j][\hat \varphi_j - \hat g \circ \hat\theta_j]\geq 0.\]
Applying the above with $g=\check g$ leads to 
\[-\sum_{j\in[J]}\P_j[(\check g - \hat g)\circ \hat\theta_j][\hat \varphi_j - \varphi]\geq -\sum_{j\in[J]}\P_j[(\check g - \hat g)\circ \hat\theta_j][\hat g \circ \hat\theta_j - \varphi], \]
which, together with  to (\ref{eq:ineq1}) yields
{\small\[\sum_{j\in[J]}||(\check g - \hat g)\circ \hat\theta_j||^2\leq -\sum_{j\in[J]}\P_j[(\check g - \hat g)\circ \hat\theta_j][\hat \varphi_j - \varphi]+ \sum_{j\in[J]}(\P_j- \P)[(\check g - \hat g)\circ \hat\theta_j][\hat g \circ \hat\theta_j - \varphi],\]}
which is equivalent to 
{\small
\begin{equation}
\begin{aligned}
\sum_{j\in[J]}||\check \theta_j - \tilde\theta_j||^2\leq \sum_{j\in[J]}\P(\check \theta_j - \tilde \theta_j)(\varphi-\hat\varphi_j) +\sum_{j\in[J]}(\P_j-\P)(\check \theta_j - \tilde\theta_j)(\tilde\theta_j - \varphi).
\end{aligned}
\end{equation}}%
In what follows we provide bounds on the terms inside the summation in the right hand side of the above expression. First, by iterated expectation and definitions we have
\begin{equation}
\begin{aligned}\sum_{j\in[J]}\P(\check \theta_j - \tilde \theta_j)(\varphi-\hat \varphi_j) &= \sum_{j\in[J]}\E_{\P}[(\check \theta_j(Z) - \tilde \theta_j(Z))\E_{\P}[\varphi(X;\eta)-\varphi(X;\hat\eta_j) \mid Z]]\\
&\leq \sum_{j\in[J]}||\check\theta_j - \tilde\theta_j||\times ||R(\cdot;\eta,\hat\eta_j)||,
\end{aligned}
\end{equation}
where the inequality follows from Cauchy-Schwarz and the fact that, by definition of $\varphi$, we have $\E[\varphi(X;\eta)\mid Z=z]=\theta(z)$ and $\E[\varphi(X;\hat\eta_j)\mid Z=z]=\theta(z) + R(z;\eta,\hat\eta_j)$, where we remind the reader that $\E[f(X)]$ is defined as an expectation with respect to $X$, keeping the function $f$ fixed.

For the second term, let $h_j = ||\check \theta_j - \tilde\theta_j||$. Let $\mathcal F_j(h) = \{[(g_1 - g_2)\circ\hat\theta_j] (g_2\circ \hat\theta_j -\varphi): g_1\in\mathcal G, g_2\in\mathcal G, ||(g_1 - g_2)\circ\hat \theta_j||\leq h\}$. Define  $S_j(h) = \sup_{f\in \mathcal F_j(h)}(\P_j-\P)f$. Then we have
\[(\P_j-\P)(\check \theta_j - \tilde\theta_j)(\tilde\theta_j -\varphi)\leq S_j(h_j).\]
We now have all the elements to proceed with the main argument, which follows a ``peeling'' argument (see the proof of Theorem 3.2.5 in \cite{vanderVaart&Wellner96}). Define $r_j =  ||R(\cdot;\eta,\hat\eta_j)||$.  Notice that the event
\[E=\left\{\sum_{j\in[J]}h_j^2\leq \sum_{j\in[J]}\{h_j r_j + S_j(h_j)\}\right\}\]
occurs with probability one. For a quantity $\varepsilon_j\geq 0 $ that is independent of $\mathcal T_j$, define the ``shells'' $N_{j,k}=\{2^k \varepsilon_j< h_j  \leq 2^{k+1}\varepsilon_j\}$ for $k=1,\ldots,\infty$. Let $M_K=\{k_1,\ldots, k_J: k_j\in \mathbb N, \max_{j\in[J]} 2^{k_j} \geq 2^K\}$. For any $K\geq 1$ we have
\begin{equation}\label{eq:jointp}
\begin{aligned}
       \Pr\left[\max_{j\in[J]}\{h_j / \varepsilon_j\}> 2^K \right] &=\sum_{M_K}\Pr\left[ N_{1,k_1},\ldots, N_{J,k_J}, E\right].    
\end{aligned}
\end{equation}
Notice that, in $N_{j,k}$, we have $h_j \leq 2^{k+1}\varepsilon_j$, and $S_j(h_j) \leq S_j(2^{k+1}\varepsilon_j)$. 
Therefore
{\small\[
\begin{aligned}
    \Pr[N_{1,k_1},\ldots, N_{J,k_K}, E]&\leq \Pr\left\{2^{k_1}\varepsilon_1 < h_1, \ldots, 2^{k_J}\varepsilon_J < h_J, \sum_{j\in[J]}h_j^2\leq \sum_{j\in[J]}\left[2^{k_j+1}\varepsilon_j r_j + S_j(2^{k_j+1}\varepsilon_j)\right]\right\}\\
    &\leq \Pr\left\{\sum_{j\in[J]}2^{2k_j}\varepsilon_j^2  < \sum_{j\in[J]}h_j^2\leq \sum_{j\in[J]}\left[2^{k_j+1}\varepsilon_j r_j + S_j(2^{k_j+1}\varepsilon_j)\right]\right\}\\
    &\leq \Pr\left\{\sum_{j\in[J]}2^{2k_j}\varepsilon_j^2  < \sum_{j\in[J]}\left[2^{k_j+1}\varepsilon_j r_j + S_j(2^{k_j+1}\varepsilon_j)\right]\right\}
\end{aligned}
\]}%
Using Markov's inequality we get this probability is  bounded above by
\begin{equation}\label{eq:ratio}
\begin{aligned}
    \sum_{j\in[J]}\frac{1}{\sum_{j\in[J]}2^{2k_j}\varepsilon_j ^2}\mathbb E\left\{2^{k_j+1}\varepsilon_j  r_j + S_j(2^{k_j+1}\varepsilon_j )\right\}.
\end{aligned}    
\end{equation}
Using Lemma~\ref{lemma:bound} below and the tower rule, we get
\[
\begin{aligned}
    \mathbb E\left\{2^{k_j+1}\varepsilon_j  r_j + S_j(2^{k_j+1}\varepsilon_j )\right\} =  \mathbb E\left\{\mathbb E\left(2^{k_j+1}\varepsilon_j  r_j + S_j(2^{k_j+1}\varepsilon_j )\,\big|\, \mathcal T_j\right)\right\}\\
   \lesssim \mathbb E\left\{ 2^{k_j+1}\varepsilon_j  r_j + \sqrt{\frac{2^{k_j+1}\varepsilon_j }{n_j}}\left(1+\frac{1}{2^{2(k_j+1)}\varepsilon_j ^2 }\sqrt{\frac{2^{k_j+1}\varepsilon_j }{n_j}}\right)\right\}.
\end{aligned}
\]
Choose $\varepsilon_j = \max\{n_j^{-1/3}, r_j\}$. For this choice, we have $r_j \leq \varepsilon_j$ and thus
\[\frac{\sum_{j\in[J]} 2^{k_j+1}\varepsilon_jr_j}{\sum_{j\in[J]}2^{2k_j}\varepsilon_j^2}\leq \frac{\sum_{j\in[J]} 2^{k_j+1}}{\sum_{j\in[J]}2^{2k_j}}=:s_1(k_1,\ldots,k_J).\]
This choice also ensures that $\varepsilon_j \geq n_j^{-1/3}$ and therefore $\sqrt{\varepsilon_j/n_j}\leq \varepsilon_j^2$, leading to
\begin{equation*}
    \frac{1}{\sum_{j\in[J]}2^{2k_j}\varepsilon ^2}\sum_{j\in [J]}\sqrt{\frac{2^{k_j+1}\varepsilon_j}{n_j}} \leq \frac{\sum_{j\in [J]}2^{(k_j+1)/2}}{\sum_{j\in[J]}2^{2k_j}}=:s_2(k_1,\ldots k_J),
\end{equation*}
as well as $n_j\varepsilon_j^3\geq 1$, leading to
\begin{equation*}
    \frac{1}{\sum_{j\in[J]}2^{2k_j}\varepsilon ^2}\sum_{j\in[J]}\frac{1}{2^{2(k_j+1)}\varepsilon_j^2 }\frac{2^{k_j+1}\varepsilon_j}{n_j} \leq \frac{\sum_{j\in[J]}2^{-k_j-1}}{\sum_{j\in[J]}2^{2k_j}}=:s_3(k_1,\ldots k_J).
\end{equation*}
Putting all of this together with (\ref{eq:ratio}) shows that 
\[\Pr\left[\max_{j\in[J]}\{h_j/\varepsilon_j\} > 2^K \right] \lesssim \sum_{M_K}\{s_1(k_1,\ldots k_J) +s_2(k_1,\ldots k_J)+ s_3(k_1,\ldots k_J)\}\xrightarrow{K\to \infty} 0,\]
where the convergence follows from Lemma \ref{lemma:conv} below. 
This implies that, for any $\epsilon>0$, we can find $K$ large enough such that $\Pr\left[\max_{j\in[J]} \{h_j / \varepsilon_j\}> 2^K\right]<\epsilon$, i.e., we have proved that $h_j = O_\Pr(\varepsilon_j)$ for all $j$. Recall the definitions $h_j = ||\check \theta_j - \tilde\theta_j||$, $\varepsilon_j = \max\{n_j^{-1/3}, r_j\}$, and  $r_j =  ||R(\cdot;\eta,\hat\eta_j)||$. Using the triangle inequality we get 
$||\tilde \theta_j - \theta||\leq ||\check \theta_j - \theta|| + ||\tilde \theta_j - \check \theta_j||$. Notice that, by definition 
\[\check\theta_j = \argmin_{g\circ\hat\theta_j:g\in\mathcal G}||g\circ\hat\theta_j-\theta||.\] 
Since the identity function is in $\mathcal G$, this means $||\check \theta_j - \theta|| \leq ||\hat \theta_j - \theta|| $.
Putting all the above together with $\varepsilon_j\leq n_j^{-1/3} + ||R(\cdot;\eta,\hat\eta_j)||$ yields the result of the theorem.
\end{proof}

\begin{lemma}\label{lemma:bound}
    Let $\mathcal V$ and $\mathcal T$ denote a partition of an i.i.d. sample  from $X\sim \P$ of size $m$ and $n-m$, respectively. Let $\P_m$ denote the empirical distribution of $\mathcal V$.  For functions $\theta:\mathcal Z\mapsto\Re$ and $\varphi:\mathcal X\mapsto\Re$ not dependent on $\mathcal V$, let $\mathcal F(h) = \{[(g_1 - g_2)\circ\theta] (g_2\circ \theta - \varphi): g_1\in\mathcal G, g_2\in\mathcal G, ||(g_1 - g_2)\circ\theta||\leq h\}$. Define $\mathbb E[f(\mathcal V)\mid \mathcal T]$ as the expectation over draws of $\mathcal V$ conditional on $\mathcal T$, and define $S(h) = \sup_{f\in \mathcal F(h)}(\P_m-\P)f$. We have
     \[\mathbb E[S(h)\mid \mathcal T]\lesssim 
     \sqrt{\frac{h}{m}}\left(1+\frac{1}{h^2}\sqrt{\frac{h}{m}}\right).\]
\end{lemma}

\begin{proof}
    This proof uses empirical process theory and definitions presented in \cite{vanderVaart&Wellner96}. For a class of functions $\mathcal F$, we let $N(\epsilon, \mathcal F, L_2(\P))$ denote the covering number defined as the minimum number of balls of radius $\epsilon$ needed to cover $\mathcal F$, where the norm $L_2(\P)$ is used to define the radius. We also define the uniform entropy integral as 
    \[J(\delta, \mathcal F, L_2(\P))=\int_0^\delta\sqrt{1 + \log N(\epsilon ||F||, \mathcal F, L_2(\P))}d\epsilon,\]
    where $F$ is an envelope of the class $\mathcal F$. 
    Note that $ N(\epsilon, \mathcal F, L_2(\P))\leq \nb(2\epsilon, \mathcal F, L_2(\P))$, where $\nb$ denotes the bracketing number. Theorem 2.7.5 of \cite{vanderVaart&Wellner96}, together with the previous fact shows that $N(\epsilon, \mathcal G, L_2(\P))\lesssim 1/\epsilon$, with $\G$ denoting the space of non-decreasing monotone functions $g:\Re \mapsto [0,1]$. As a consequence, using $\sqrt{a+b}\leq \sqrt{a} + \sqrt{b}$ we get $J(\delta, \mathcal G, L_2(\P))\lesssim \delta^{1/2}$. 
    
    Let $M$ be such that $|\varphi(x)|\leq M$. Define the class $\mathcal F'(h) = \{[(g_1 - g_2)\circ\theta] (g_2\circ \theta - \varphi) /(1+M) : g_1\in\mathcal G, g_2\in\mathcal G, ||(g_1 - g_2)\circ\theta||\leq h\}$. As a result of Theorem 2.10.20 of \cite{vanderVaart&Wellner96}, we have $J(\delta, \mathcal F'(h), L_2(\P))\lesssim \delta^{1/2}$. Note that an envelope for $\mathcal F'(h)$ is $F=1$. Thus, an application of Theorem 2.1 in \cite{van2011local} with $\delta=h$ yields the result of the lemma.
\end{proof}

\begin{lemma}\label{lemma:conv}
    For $K\geq 1$, let $M_K=\{k_1,\ldots, k_J: k_j\in \mathbb N, \max_{j\in[J]} 2^{k_j} \geq 2^K\}$, and let $f$ and $g$ denote functions such that $g(k) > f(k)$ and $g'(k) > f'(k)$. Define
    \[h(M_K)=\sum_{M_K}\frac{\sum_{j\in[J]}2^{f(k_j)}}{\sum_{j\in[J]}2^{g(k_j)}}.\]
    We have $\lim_{K\to\infty}h(M_K)=0.$
\end{lemma}
\begin{proof}
Notice that     
\begin{equation}\label{eq:sum1}
    h(M_K)\lesssim \sum_{(k_2,\ldots,k_J):k_j\geq1}\sum_{k_1\geq K}\frac{2^{f(k_1)} + \cdots + 2^{f(k_J)}}{2^{g(k_1)} + \cdots + 2^{g(k_J)}} = \sum_{(k_1,\ldots,k_J):k_j\geq1}\frac{2^{f(k_1+K)} + \cdots + 2^{f(k_J)}}{2^{g(k_1+K)} + \cdots + 2^{g(k_J)}}. 
\end{equation}
Furthermore, applying L'H\^{o}pital's rule notice that 
\[1\geq\frac{2^{f(k_1+K)} + \cdots + 2^{f(k_J)}}{2^{g(k_1+K)} + \cdots + 2^{g(k_J)}}\xrightarrow[]{K\to\infty} 0.\]
Thus, the dominated convergence theorem can be used to 
take the limit in (\ref{eq:sum1}), interchanging the sum and limit signs, and yielding the result of the lemma.
\end{proof}

\section{Proof of Theorem 1}
We provide a sketch of the proof of asymptotic normality as it follows a similar structure as previous results, such as the proof of \citealt[Theorem 3]{diaz2023nonparametric}. The principle difference is in the application of the isotonic calibration step, which requires applying our Lemma 1.

To see how our Lemma 1 applies, note first that by \citealt[Lemma 1]{diaz2023nonparametric}, the parameter $m_{t+1, s}$ satisfies the following expansion for any $\eta'$:
\begin{align}
    m_{t+1,s}(a_s, h_s) = \E\left[ \varphi_{t+1,s+1}(Z, \eta') \mid A_s = a_s, H_s = h_s \right] + \Rem_s(a_s, h_s; \eta'),
\end{align}
where the second-order remainder term $\Rem_s$ is as defined in the citation.
Therefore, by Lemma 1, we have that
\begin{align}
    \| \tilde{m}_{t+1, s, j} - m_{t+1,s} \| \lesssim \| \hat{m}_{t + 1, s, j} - m_{t+1,s} \| + O_{\mathbb{P}}\left( n^{-1/3} + \|\Rem_s(\cdot; \eta, \hat{\eta}_j)\| \right)
\end{align}
Next, we apply the usual decomposition of the estimator into a central limit term, empirical process term, and bias term.
Recall that for $l = 1, \dots, \tau$, the pooled SDR estimator is given by
\begin{align}
    \hat{\theta}_{\mathrm{sdr}}(l+1) &= \frac{1}{n} \sum_{i=1}^n \varphi_{l+1,1}\left( Z_i, \tilde{\eta}_{j(i)} \right) \\
    &= \frac{1}{J} \sum_{j=1}^J \P_{n,j} \varphi_{l+1,1}(\tilde{\eta}_j).
\end{align}
Next, write
\begin{align}
    \sqrt{n} \left( \hat{\theta}_{\mathrm{sdr}}(l+1) - \theta(l+1) \right) = \mathsf{G}_n(\varphi_{l+1,1}(\eta) - \theta) + R_{n,1} + R_{n,2},
\end{align}
where $R_{n,1}$ is called the empirical process term andf $R_{n,2}$ the bias term, given by
\begin{align}
    R_{n,1} &= \frac{1}{\sqrt{J}} \sum_{j=1}^J \mathsf{G}_{n,j}\left( \varphi_{l+1,1}(\tilde{\eta}_j) - \varphi_{l+1,1}(\eta) \right), \\
    R_{n,2} &= \frac{\sqrt{n}}{J} \sum_{j=1}^J \mathsf{P}(\varphi_{l+1, 1}(\tilde{\eta}_j) - \theta(l + 1)).
\end{align}
To handle the bias term $R_{n,2}$, note that by \citealt[Lemma 1]{diaz2023nonparametric}, for any $\eta'$, $ \P[\varphi_{l+1,1}(\eta')- \theta(l +1)] = \Rem_0(\eta')$, and 
by \citealt[Lemma 3]{diaz2023nonparametric} (and under the assumptions of our theorem), then
\begin{align}
    \Rem_0(\tilde{\eta}) = \sum_{t=1}^{l+1} O_{P} \left( \| \hat{r}_t - r_t \| \| \tilde{m}_{l+1,t} - \tilde{m}_{l+1,t}^\dagger \| \right),
\end{align}
where $\tilde{m}_{l+1,t}^\dagger(a_t, h_t) = \E\left[ \varphi_{l+1, t+1}(\underline{\tilde{\eta}_t} \mid A_t = a_t, H_t = h_t \right]$ (that is, $\tilde{m}_t^\dagger$ is the expectation over the distribution $\P$ with nuisances fixed to $\tilde{\eta}$). Plugging in the results from our Lemma 1 implies that the remainder term becomes (suppressing the cross-fitting notation)
\begin{align}
     \Rem_0(\tilde{\eta}) = \sum_{t=1}^{l+1} O_{P} \left( \| \hat{r}_t - r_t \| \left( \| \hat{m}_{l + 1, t} - m_{l+1,t} \| + O_{\mathbb{P}}\left( n^{-1/3} + \|\Rem_t(\cdot; \eta, \hat{\eta})\| \right) \right)\right),
\end{align}
Therefore the bias term converges as $o_P(1)$ under the conditions of the theorem; in other words, the isotonic correction step introduces a term that converges faster than $n^{-1/2}$, therefore not affecting the convergence of the bias term, and we are able to establish asymptotic normality by placing conditions on the original nuisance estimators $\hat{m}$.  
For the empirical process term, the cross-fitting of the nuisance estimators, combined with the results of our Lemma 1, imply that the $R_{n,1}$ is also $o_P(1)$. This leaves only the central limit term, which establishes asymptotic normality. The convergence of the vector of estimators to a joint normal distribution follows from e.g. \citealt[Theorem 18.10]{van2002part}. The local asymptotic minimax efficiency bound statement is derived from \citealt{van2002part}. 


\section{Additional Simulation Results}

\subsection{Simulation Study 1}

Additional results for Simulation Study 1 for the empirical coverage of the pointwise 95\% confidence intervals and uniform 95\% uniform confidence bands are shown in Table~\ref{tab:simulation-study-coverage-1}.

\begin{table}[h]
    \centering
    \begin{tabular}{|lrrrrr|}
    \hline
    & & \multicolumn{2}{c}{Pointwise 95\% Coverage} & \multicolumn{2}{c|}{Uniform 95\% Coverage} \\
    $\alpha$ & $n$ & SDR Curve & Benchmark & SDR Curve & Benchmark \\
    \hline
    0 & 500 & 94.25\% & 94.20\% & 96.85\% & 96.60\%\\
     & 1000 & 95.00\% & 94.45\% & 97.25\% & 97.10\%\\
     & 2500 & 94.20\% & 94.15\% & 96.80\% & 97.00\%\\
     & 5000 & 94.10\% & 93.65\% & 96.55\% & 96.50\%\\
    0.5 & 500 & 92.95\% & 93.70\% & 95.50\% & 96.05\%\\
     & 1000 & 94.00\% & 94.30\% & 96.10\% & 96.25\%\\
     & 2500 & 94.55\% & 95.45\% & 97.20\% & 97.60\%\\
     & 5000 & 94.95\% & 95.30\% & 97.50\% & 97.25\%\\
    0.8 & 500 & 90.50\% & 91.70\% & 93.30\% & 93.90\%\\
     & 1000 & 93.55\% & 94.15\% & 95.75\% & 96.50\%\\
     & 2500 & 92.80\% & 93.30\% & 96.15\% & 96.05\%\\
     & 5000 & 94.10\% & 93.70\% & 96.55\% & 96.75\%\\
    \hline
    \end{tabular}
    \caption{Results of Simulation Study~1 comparing the proposed time-smoothed SDR algorithm (SDR curve) against the benchmark in terms of the 95\% empirical coverage of the pointwise and uniform confidence regions. }
    \label{tab:simulation-study-coverage-1}
\end{table}

\clearpage

\subsection{Simulation Study 2}

Additional results for Simulation Study 2 for the empirical coverage of the pointwise 95\% confidence intervals and uniform 95\% uniform confidence bands are shown in Table~\ref{tab:simulation-study-coverage-2}.

\begin{table}[h]
    \centering
    \begin{tabular}{|lrrrrrrr|}
    \hline
    & & \multicolumn{3}{c}{95\%  Coverage} & \multicolumn{3}{c|}{95\% Uniform Coverage} \\
    & & Unconstrained & & & Unconstrained & & \\
    $\alpha$ & $n$ & SDR Curve & SDR Curve & Benchmark & SDR Curve  & SDR Curve & Benchmark \\
    \hline
    0 & 500 & 90.60\% & 91.10\% & 90.80\% & 92.80\% & 93.20\% & 93.35\%\\
     & 1000 & 91.75\% & 93.80\% & 93.35\% & 94.20\% & 95.90\% & 95.30\%\\
     & 2500 & 95.00\% & 96.15\% & 95.40\% & 97.25\% & 97.90\% & 97.50\%\\
     & 5000 & 94.75\% & 96.60\% & 95.70\% & 97.00\% & 98.10\% & 97.75\%\\
    1.5 & 500 & 93.60\% & 90.90\% & 93.00\% & 96.65\% & 93.85\% & 95.70\%\\
     & 1000 & 94.85\% & 93.50\% & 95.15\% & 98.00\% & 96.15\% & 97.50\%\\
     & 2500 & 95.55\% & 94.50\% & 96.35\% & 98.85\% & 97.30\% & 98.10\%\\
     & 5000 & 95.15\% & 94.00\% & 95.95\% & 97.80\% & 96.35\% & 97.80\%\\
    3 & 500 & 94.00\% & 90.35\% & 93.35\% & 97.05\% & 93.60\% & 96.40\%\\
     & 1000 & 93.85\% & 91.35\% & 94.30\% & 97.75\% & 94.35\% & 97.10\%\\
     & 2500 & 94.65\% & 89.65\% & 94.55\% & 98.20\% & 93.45\% & 97.10\%\\
     & 5000 & 93.90\% & 89.65\% & 93.55\% & 96.60\% & 92.95\% & 96.20\%\\
    \hline
    \end{tabular}
    \caption{Results of Simulation Study 2 comparing the proposed time-smoothed SDR algorithm (SDR Curve), the time-smoothed SDR algorithm with the isotonic constraint step removed (Unconstrained SDR Curve); and the benchmark in terms of the 95\% empirical coverage of the pointwise and uniform confidence regions.}
    \label{tab:simulation-study-coverage-2}
\end{table}

\clearpage

\bibliographystyle{plainnat}
\bibliography{refs}

\begin{thebibliography}{34}
\providecommand{\natexlab}[1]{#1}
\providecommand{\url}[1]{\texttt{#1}}
\expandafter\ifx\csname urlstyle\endcsname\relax
  \providecommand{\doi}[1]{doi: #1}\else
  \providecommand{\doi}{doi: \begingroup \urlstyle{rm}\Url}\fi

\bibitem[Belloni et~al.(2018)Belloni, Chernozhukov, Chetverikov, and
  Wei]{belloni2018inference}
Alexandre Belloni, Victor Chernozhukov, Denis Chetverikov, and Ying Wei.
\newblock {Uniformly valid post-regularization confidence regions for many
  functional parameters in z-estimation framework}.
\newblock \emph{The Annals of Statistics}, 46\penalty0 (6B):\penalty0 3643 --
  3675, 2018.
\newblock \doi{10.1214/17-AOS1671}.
\newblock URL \url{https://doi.org/10.1214/17-AOS1671}.

\bibitem[Cain et~al.(2010)Cain, Robins, Lanoy, Logan, Costagliola, and
  Hern{\'a}n]{cain2010start}
Lauren~E Cain, James~M Robins, Emilie Lanoy, Roger Logan, Dominique
  Costagliola, and Miguel~A Hern{\'a}n.
\newblock When to start treatment? a systematic approach to the comparison of
  dynamic regimes using observational data.
\newblock \emph{The international journal of biostatistics}, 6\penalty0 (2),
  2010.

\bibitem[Chernozhukov et~al.(2013)Chernozhukov, Chetverikov, Kato,
  et~al.]{chernozhukov2013gaussian}
Victor Chernozhukov, Denis Chetverikov, Kengo Kato, et~al.
\newblock Gaussian approximations and multiplier bootstrap for maxima of sums
  of high-dimensional random vectors.
\newblock \emph{The Annals of Statistics}, 41\penalty0 (6):\penalty0
  2786--2819, 2013.

\bibitem[Chernozhukov et~al.(2016)Chernozhukov, Chetverikov, Demirer, Duflo,
  Hansen, et~al.]{chernozhukov2016double}
Victor Chernozhukov, Denis Chetverikov, Mert Demirer, Esther Duflo, Christian
  Hansen, et~al.
\newblock Double machine learning for treatment and causal parameters.
\newblock \emph{arXiv preprint arXiv:1608.00060}, 2016.

\bibitem[D{\'\i}az et~al.(2023)D{\'\i}az, Williams, Hoffman, and
  Schenck]{diaz2023nonparametric}
Iv{\'a}n D{\'\i}az, Nicholas Williams, Katherine~L Hoffman, and Edward~J
  Schenck.
\newblock Nonparametric causal effects based on longitudinal modified treatment
  policies.
\newblock \emph{Journal of the American Statistical Association}, 118\penalty0
  (542):\penalty0 846--857, 2023.

\bibitem[D{\'\i}az et~al.(2024)D{\'\i}az, Hoffman, and Hejazi]{diaz2024causal}
Iv{\'a}n D{\'\i}az, Katherine~L Hoffman, and Nima~S Hejazi.
\newblock Causal survival analysis under competing risks using longitudinal
  modified treatment policies.
\newblock \emph{Lifetime Data Analysis}, 30\penalty0 (1):\penalty0 213--236,
  2024.

\bibitem[Gin{\'e} and Zinn(1984)]{gine1984some}
Evarist Gin{\'e} and Joel Zinn.
\newblock Some limit theorems for empirical processes.
\newblock \emph{The Annals of Probability}, pages 929--989, 1984.

\bibitem[Haneuse and Rotnitzky(2013)]{Haneuse2013}
Sebastian Haneuse and Andrea Rotnitzky.
\newblock Estimation of the effect of interventions that modify the received
  treatment.
\newblock \emph{Statistics in Medicine}, 2013.

\bibitem[Hern{\'a}n et~al.(2002)Hern{\'a}n, Brumback, and
  Robins]{hernanrepeated}
Miguel~A Hern{\'a}n, Babette~A Brumback, and James~M Robins.
\newblock Estimating the causal effect of zidovudine on cd4 count with a
  marginal structural model for repeated measures.
\newblock \emph{Statistics in Medicine}, 21\penalty0 (12):\penalty0 1689--709,
  2002.

\bibitem[Hern{\'a}n et~al.(2009)Hern{\'a}n, McAdams, McGrath, Lanoy, and
  Costagliola]{hernanobsplans}
Miguel~A Hern{\'a}n, Mara McAdams, Nuala McGrath, Emilie Lanoy, and Dominique
  Costagliola.
\newblock Observation plans in longitudinal studies with time-varying
  treatments.
\newblock \emph{Stat Methods Med Res}, 18\penalty0 (1):\penalty0 27 -- 52,
  2009.

\bibitem[Hu and Hogan(2019)]{hu2019causal}
Liangyuan Hu and Joseph~W Hogan.
\newblock Causal comparative effectiveness analysis of dynamic continuous-time
  treatment initiation rules with sparsely measured outcomes and death.
\newblock \emph{Biometrics}, 75\penalty0 (2):\penalty0 695--707, 2019.

\bibitem[Luedtke et~al.(2017)Luedtke, Sofrygin, van~der Laan, and
  Carone]{luedtke2017sequential}
Alexander~R Luedtke, Oleg Sofrygin, Mark~J van~der Laan, and Marco Carone.
\newblock Sequential double robustness in right-censored longitudinal models.
\newblock \emph{arXiv preprint arXiv:1705.02459}, 2017.

\bibitem[McGrath et~al.(2020)McGrath, Lin, Zhang, Petito, Logan, Hern{\'a}n,
  and Young]{mcgrath2020gformula}
Sean McGrath, Victoria Lin, Zilu Zhang, Lucia~C Petito, Roger~W Logan, Miguel~A
  Hern{\'a}n, and Jessica~G Young.
\newblock gformula: an r package for estimating the effects of sustained
  treatment strategies via the parametric g-formula.
\newblock \emph{Patterns}, 1\penalty0 (3), 2020.

\bibitem[McGrath et~al.(2025)McGrath, Kawahara, Petimar, Rifas-Shiman,
  D{\'\i}az, Block, and Young]{mcgrathrepeated}
Sean McGrath, Takuya Kawahara, Joshua Petimar, Sherl~L. Rifas-Shiman, Iv{\'a}n
  D{\'\i}az, Jason~P. Block, and Jessica~G. Young.
\newblock Time-smoothed inverse probability weighted estimation of effects of
  generalized time-varying treatment strategies on repeated outcomes truncated
  by death.
\newblock \emph{arXiv preprint arXiv:2509.13971}, 2025.

\bibitem[Pfanzagl and Wefelmeyer(1985)]{pfanzagl1982contributions}
J~Pfanzagl and W~Wefelmeyer.
\newblock Contributions to a general asymptotic statistical theory.
\newblock \emph{Statistics \& Risk Modeling}, 3\penalty0 (3-4):\penalty0
  379--388, 1985.

\bibitem[Richardson and Robins(2013)]{richardson2013single}
Thomas~S Richardson and James~M Robins.
\newblock Single world intervention graphs ({SWIG}s): A unification of the
  counterfactual and graphical approaches to causality.
\newblock \emph{Center for the Statistics and the Social Sciences, University
  of Washington Series. Working Paper}, 128\penalty0 (30):\penalty0 2013, 2013.

\bibitem[Robins(1986)]{Robins86}
James~M Robins.
\newblock A new approach to causal inference in mortality studies with
  sustained exposure periods - application to control of the healthy worker
  survivor effect.
\newblock \emph{Mathematical Modelling}, 7:\penalty0 1393--1512, 1986.

\bibitem[Robins et~al.(1995)Robins, Rotnitzky, and Zhao]{robins1995analysis}
James~M Robins, Andrea Rotnitzky, and Lue~Ping Zhao.
\newblock Analysis of semiparametric regression models for repeated outcomes in
  the presence of missing data.
\newblock \emph{Journal of the american statistical association}, 90\penalty0
  (429):\penalty0 106--121, 1995.

\bibitem[Robins et~al.(2004)Robins, Hern{\'a}n, and Siebert]{robins2004effects}
James~M Robins, Miguel~A Hern{\'a}n, and Uwe Siebert.
\newblock Effects of multiple interventions.
\newblock \emph{Comparative quantification of health risks: global and regional
  burden of disease attributable to selected major risk factors}, 1:\penalty0
  2191--2230, 2004.

\bibitem[Rotnitzky et~al.(2017)Rotnitzky, Robins, and
  Babino]{rotnitzky2017multiply}
Andrea Rotnitzky, James Robins, and Lucia Babino.
\newblock On the multiply robust estimation of the mean of the g-functional.
\newblock \emph{arXiv preprint arXiv:1705.08582}, 2017.

\bibitem[Shahu and Malinsky(2025)]{shahu2025estimating}
Anja Shahu and Daniel Malinsky.
\newblock Estimating effects of longitudinal modified treatment policies
  (lmtps) on rates of change in health outcomes, 2025.

\bibitem[Shea(2024)]{shea2024wooldridge}
Justin~M. Shea.
\newblock \emph{wooldridge: 115 Data Sets from "Introductory Econometrics: A
  Modern Approach, 7e" by Jeffrey M. Wooldridge}, 2024.
\newblock URL \url{https://CRAN.R-project.org/package=wooldridge}.
\newblock R package version 1.4-4.

\bibitem[Shi et~al.(2024)Shi, Ke, Soukhavong, Lamb, Meng, Finley, Wang, Chen,
  Ma, Ye, Liu, Titov, and Cortes]{shi2024lightgbm}
Yu~Shi, Guolin Ke, Damien Soukhavong, James Lamb, Qi~Meng, Thomas Finley,
  Taifeng Wang, Wei Chen, Weidong Ma, Qiwei Ye, Tie-Yan Liu, Nikita Titov, and
  David Cortes.
\newblock \emph{lightgbm: Light Gradient Boosting Machine}, 2024.
\newblock URL \url{https://CRAN.R-project.org/package=lightgbm}.
\newblock R package version 4.4.0.

\bibitem[van~der Laan et~al.(2023)van~der Laan, Ulloa-Perez, Carone, and
  Luedtke]{pmlr-v202-van-der-laan23a}
Lars van~der Laan, Ernesto Ulloa-Perez, Marco Carone, and Alex Luedtke.
\newblock Causal isotonic calibration for heterogeneous treatment effects.
\newblock In Andreas Krause, Emma Brunskill, Kyunghyun Cho, Barbara Engelhardt,
  Sivan Sabato, and Jonathan Scarlett, editors, \emph{Proceedings of the 40th
  International Conference on Machine Learning}, volume 202 of
  \emph{Proceedings of Machine Learning Research}, pages 34831--34854. PMLR,
  23--29 Jul 2023.
\newblock URL \url{https://proceedings.mlr.press/v202/van-der-laan23a.html}.

\bibitem[{van der Laan} and Rose(2011)]{vanderLaanRose11}
Mark~J {van der Laan} and Sherri Rose.
\newblock \emph{Targeted Learning: Causal Inference for Observational and
  Experimental Data}.
\newblock Springer, New York, 2011.

\bibitem[van~der Laan et~al.(2007)van~der Laan, Polley, and
  Hubbard]{vdl2007superlearner}
Mark~J. van~der Laan, Eric~C Polley, and Alan~E. Hubbard.
\newblock Super learner.
\newblock \emph{Statistical Applications in Genetics and Molecular Biology},
  6\penalty0 (1), 2007.
\newblock \doi{doi:10.2202/1544-6115.1309}.
\newblock URL \url{https://doi.org/10.2202/1544-6115.1309}.

\bibitem[van~der Vaart(2002)]{van2002part}
Aad van~der Vaart.
\newblock Semiparameric statistics.
\newblock \emph{Lectures on Probability Theory and Statistics}, pages 331--457,
  2002.

\bibitem[van~der Vaart and Wellner(2011)]{van2011local}
Aad van~der Vaart and Jon~A Wellner.
\newblock A local maximal inequality under uniform entropy.
\newblock \emph{Electronic Journal of Statistics}, 5\penalty0 (2011):\penalty0
  192, 2011.

\bibitem[{van der Vaart} and Wellner(1996)]{vanderVaart&Wellner96}
Aad~W {van der Vaart} and Jon~A Wellner.
\newblock \emph{Weak {C}onvergence and {E}mprical {P}rocesses}.
\newblock Springer-Verlag New York, 1996.

\bibitem[Vella and Verbeek(1998)]{vella1998wages}
Francis Vella and Marno Verbeek.
\newblock Whose wages do unions raise? a dynamic model of unionism and wage
  rate determination for young men.
\newblock \emph{Journal of Applied Econometrics}, 13\penalty0 (2):\penalty0
  163--183, 1998.
\newblock ISSN 08837252, 10991255.
\newblock URL \url{http://www.jstor.org/stable/223257}.

\bibitem[Williams and Díaz(2023)]{williams2023lmtp}
Nicholas Williams and Iván Díaz.
\newblock lmtp: An r package for estimating the causal effects of modified
  treatment policies.
\newblock \emph{Observational Studies}, 2023.
\newblock URL \url{https://muse.jhu.edu/article/883479}.

\bibitem[Wright and Ziegler(2017)]{wright2017ranger}
Marvin~N. Wright and Andreas Ziegler.
\newblock {ranger}: A fast implementation of random forests for high
  dimensional data in {C++} and {R}.
\newblock \emph{Journal of Statistical Software}, 77\penalty0 (1):\penalty0
  1--17, 2017.
\newblock \doi{10.18637/jss.v077.i01}.

\bibitem[Young et~al.(2014)Young, Hern{\'a}n, and
  Robins]{young2014identification}
Jessica~G Young, Miguel~A Hern{\'a}n, and James~M Robins.
\newblock Identification, estimation and approximation of risk under
  interventions that depend on the natural value of treatment using
  observational data.
\newblock \emph{Epidemiologic methods}, 3\penalty0 (1):\penalty0 1--19, 2014.

\bibitem[Young et~al.(2019)Young, Logan, Robins, and
  Hern{\'a}n]{young2018inverse}
Jessica~G Young, Roger~W Logan, James~M Robins, and Miguel~A Hern{\'a}n.
\newblock Inverse probability weighted estimation of risk under representative
  interventions in observational studies.
\newblock \emph{Journal of the American Statistical Association}, 114\penalty0
  (526):\penalty0 938 -- 947, 2019.

\end{thebibliography}
\end{document}